\documentclass[11pt]{amsart}
\usepackage{amssymb,latexsym, amsmath, amsxtra}
\usepackage[dvips]{graphics}
\usepackage{epsfig}

\usepackage[usenames, dvipsnames]{color} 
\usepackage[table]{xcolor}  
\usepackage{amssymb}
\usepackage{amsmath}
\usepackage{tikz}
\usepackage{hyperref}
\usepackage{amsthm}
\usepackage{mathtools}
\usepackage{braket}
\usepackage{physics}
\usepackage{bbm}
\usepackage{stmaryrd}
\usepackage{float}
\usepackage{graphicx}
\usepackage{times}
\usepackage{mdframed}
\usepackage{caption}
\usepackage{subcaption}
\usetikzlibrary{backgrounds}
\usepackage{empheq}
\usepackage{MnSymbol,wasysym}
\usepackage{mleftright}
\usepackage{cleveref}
\usetikzlibrary{arrows,matrix,positioning}

\usetikzlibrary{calc,shapes}

\newtheorem{theorem}{Theorem}[section]

\newcommand{\Zprod}{Z}

\begin{document}
\title{Averages of traces and ratios of characteristic polynomials of random unitary matrices}
\author{A. Bendel, R. Cahill, J.B. Conrey, G. Harrison Barraclough, A. Irfan,  N.C. Snaith and L. Tyulenev}
\thanks{Thanks to the American Institute of Mathematics for hosting the work at the initiation of this project.  AB was supported by the London Mathematical Society on an Undergraduate Research Bursary.  The fifth author was supported by the Heilbronn Institute for Mathematical Research. LT was funded on a summer internship by the School of Mathematics, University of Bristol.  Thanks to Ifeoluwa Oyinlola and Huw Day for their preliminary numerical work on trace averages. }

\begin{abstract} In this work we explore the extent to which the Ratios Theorem for averages over $U(N)$, with respect to Haar measure, of ratios of characteristic polynomials can be used to evaluate averages of traces of unitary matrices. We derive concrete expressions for all averages involving up to 4 traces, valid for any matrix size, $N$.  We also calculate the so-called zero-swap and one-swap contributions to an average with any number of traces, which amounts to determining general trace moments in a range double the size of the renowned result of Diaconis and Shahshahani. The expression is tractable  enough to use in applications, for example to determine the parameter range where certain trace averages assume negative values. 
\end{abstract}
\maketitle

\section{Introduction}

In 1994, the seminal result of Diaconis and Shahshahani \cite{kn:diasha} partially solved the question of averages of powers of traces of unitary matrices using an elegant method from representation theory.  Averages of traces appear in many active research directions, as is seen in the diversity of work citing Diaconis and Shahshahani.  There are multiple links to number theory; for example, these averages can relate to counting almost-primes in short intervals \cite{kn:bcis,kn:kuplal} and provide insight into the analysis of the Riemann zeta function, as well as related $L$-functions. They also feature in quantum physics (eg  \cite{kn:roth18,kn:doosti22,kn:necpar25}), where quantum quenches use them to characterise how complex many-body systems thermalise, and in quantum information, where the long-term spectrum of states produced by quantum channels is governed by a trace average over unitary matrices. 

However, the elegant result of Diaconis and Shahshahani holds only in a specific range when the matrix size is large enough. 

The compact space of unitary matrices of dimension $N$, $U(N)$, admits a  measure (Haar measure) that weights all matrices equally.
  If the eigenvalues of a unitary matrix, $U$, are $e^{i\theta_1},\ldots,e^{i\theta_N}$, then the trace of the  matrix is the sum of the eigenvalues
$${\rm Tr}\: U=e^{i\theta_1}+\cdots +e^{i\theta_N}\;{\rm and}\;{\rm Tr}\:(U^m)=e^{im\theta_1}+\cdots +e^{im\theta_N},$$
where $m$ is a positive integer. 

The theorem of Diaconis and Shahshahani (see \cite{kn:diasha} for the original paper, but also \cite{kn:diaeva01} which corrects an error in the statement of the theorem) states that for vectors ${\mathbf a}=(a_1,\ldots,a_l)$ and ${\mathbf b}=(b_1,\ldots,b_l)$, with $a_j,b_j$ non-negative integers, an average (over all possible $N\times N$ unitary matrices) of combinations of traces of powers of matrices takes the following form (here $\delta_{\mathbf{ab}}$ is one if the vectors  ${\mathbf a}$ and  ${\mathbf b}$ are equal, and zero otherwise)
\begin{equation}\label{eq:DS}
    \int_{U(N)} \prod_{j=1}^l ({\rm Tr} M^j)^{a_j}\overline{({\mathrm Tr}M^j)}^{b_j} dU = \delta_{\mathbf{ab}} \prod_{j=1}^l j^{a_j} a_j!,
\end{equation}
but only when $N\geq {\rm max} \big(\sum_{j=1}^l ja_j,\sum_{j=1}^l jb_j\big)$.  That is, for large enough matrix size $N$, the traces of random unitary matrices behave essentially like independent  Gaussian complex random variables.  However, this does not hold outside the range specified and outside that range  the problem becomes much more intractable.  

There have been some advances beyond this, with Rains \cite{kn:rains97} giving some insight into the problem if $a_j=b_j=0$ when $j<N$. Dyson, in fact, has early work that implies the result for averages of two traces ($a_j=b_j=1$ for one value of $j$, with the other $a$'s and $b$'s being zero) for any $N$. This is not completely explicit in Dyson's paper, but follows from equation (82) of \cite{kn:dyson3}, and is further explained at equation (3.5) of \cite{kn:kea05}. Pastur and Vasilchuk \cite{kn:pasvas04} calculated a specific average in the case that only one of the $a_j$ and one of the $b_j$ are non-zero, that holds in a range of small $N$ outside of the range of validity of Diaconis-Shahshahani. Kuperberg and Lal{\'i}n \cite{kn:kuplal}, in Section 7.3, consider averages of four traces in a range where $2N+1\geq {\rm max} \big(\sum_{j=1}^l ja_j,\sum_{j=1}^l jb_j\big)$. They also use the Ratio Theorem, to derive their results. It is also known, see for example \cite{kn:tao19}, that the averages of traces are zero unless the sum of the exponents on the traces of $U$ and the sum of the exponents on the complex conjugates (sometimes written as $\overline{\Tr M^j}$ or else $\Tr((M^\dagger)^j)$)  are equal:  
\begin{equation}\label{eq:balancing}
\sum_{j}ja_j=\sum_{j}jb_j.
\end{equation}
  Beyond this, much literature has concentrated on the large $N$ limit and the rate of approach to Gaussian moments (see for example \cite{kn:johan97,kn:johlam,kn:cjl24,kn:Irfan_traces}).  Recently, Irfan has provided concrete formulae for averages of traces for any matrix size \cite{kn:Irfan_traces}.  Weingarten Calculus (see \cite{kn:CMN22,kn:consni}, but particularly \cite{kn:col03}, for example in Proposition 3.2) also provides a way of writing down formulae for averages of traces. Soshnikov and Wu look instead at cumulants of traces in \cite{kn:soswu}. These formulae fully solve the problem via the moment-cumulant formula, but highlight the intense complexity of the combinatorics of the trace averages. 

In this paper, we use the fact  that averages, with respect to Haar measure, of logarithmic derivatives of characteristic polynomials of matrices from $U(N)$ are generating functions for the trace averages and we utilise the formula for averages of logarithmic derivatives derived from the Ratios Theorem (an expression for averages of ratios of characteristic polynomials of random matrices) to explore trace averages beyond the Diaconis-Shahshahani region in certain cases where we can obtain explicit expressions for averages of traces that have an informative structure. The Ratios Theorem has already proved so useful - yielding  information about eigenvalue statistics, moments of characteristic polynomials and their derivatives - that it seems worthwhile to explore what the Ratios Theorem can tell us about averages of traces, seeking structure that might not be apparent from previous methods.  

Our main theorems are: Theorem \ref{thrm: 1--n trace moment}, for the case of one trace of a power of $U^\dagger$  and any number of traces of powers of $U$, valid for any $N$; Theorem \ref{theo:2-2full}, for the case of two traces of powers of $U^\dagger$ and two traces of powers of $U$, also valid for any $N$; and Theorem \ref{theo:generalfull} for any general trace average where the sum of the exponents on the $U$'s is less than $2N+2$.

\section{Ratios Theorem and the logarithmic derivative of the characteristic polynomial}

We define the characteristic polynomial of a matrix $U\in U(N)$, with eigenvalues $e^{i\theta_1},\ldots,e^{i\theta_N}$,  as 
\begin{equation}
\Lambda_U(s)=\prod_{j=1}^N (1-se^{-i\theta_j}).
\end{equation}

The Ratios Theorem, first proposed by Conrey, Farmer and Zirnbauer in Theorem 4.1 of \cite{kn:cfz2}, revisited by Conrey, Forrester and Snaith in Proposition 2.1 of \cite{kn:cfs05}, and written in a slightly more intuitive set notation in Theorem 3.2 of \cite{kn:consna08}, gives a formula for ratios of products of characteristic polynomials, each evaluated at different points in the complex plane, of matrices drawn from  $U(N)$ with Haar measure. By differentiating with respect to the points of evaluation of the characteristic polynomials in the numerator, an expression for averages of logarithmic derivatives is obtained \cite{kn:consna08}.  We restate this Theorem here. 
\begin{theorem}[Averages of logarithmic derivatives: Theorem 3.3 \cite{kn:consna08}]
\label{theo:J} If $\Re a_j>0$ and $ \Re b_j
>0$ for  $a_j\in {\mathcal A}$ and $b_j\in {\mathcal B}$ and $U^\dagger$ is the conjugate transpose of $U$, then $J({\mathcal A};{\mathcal B})=J^*({\mathcal A};{\mathcal B})$  where
\begin{eqnarray}\label{eq:LDmom}
J({\mathcal A};{\mathcal B})&:=& \int_{U(N)}\prod_{a\in {\mathcal A}}
(-e^{-a})\frac{\Lambda_U'}{\Lambda_U}(e^{-a})\prod_{b\in
{\mathcal B}} (-e^{-b})\frac{\Lambda_{U^\dagger}'}{\Lambda_{U^\dagger}}(e^{-b})~ dU
,
\end{eqnarray}
 \begin{eqnarray} \label{eq:J*} &&J^*({\mathcal A};{\mathcal B}):= \nonumber \\
 &&\qquad\qquad\sum_{S\subset {\mathcal A},T\subset {\mathcal B}\atop
|S|=|T|}e^{-N(\sum_{\hat a\in S} \hat a
+\sum_{\hat{b}\in T}\hat b)} \frac{Z(S,T)Z(S^-,T^-)} {
Z^{\dagger}(S,S^-)Z^{\dagger}(T,T^-)} \sum_{{({\mathcal A}-S)+ ({\mathcal B}-T)\atop =
U_1+\dots + U_R}\atop |U_r|\le 2}\prod_{r=1}^R H_{S,T}(U_r),
\end{eqnarray}
and
\begin{equation}\label{eqn:Hrmt}
H_{S,T}(W)=\left\{\begin{array}{ll} \sum_{\hat a\in
S}\frac{z'}{z}(a-\hat{a})-\sum_{\hat b \in T}
\frac{z'}{z}(a +\hat b) &\mbox{ if $W=\{a\}\subset
{\mathcal A}-S$}
   \\
\sum_{\hat b \in T}\frac{z'}{z}(b-\hat
b)-\sum_{\hat a\in S} \frac{z'}{z}
(b+\hat a) &\mbox{ if  $W=\{b\}\subset {\mathcal B}-T$}\\
\left(\frac{z'}{z}\right)'(a+b) & \mbox{ if
$W=\{a,b\}$ with $
{a \in {\mathcal A}-S, \atop b\in {\mathcal B}-T}$}\\
0&\mbox{ otherwise}.
\end{array}\right.
\end{equation}
The definition of $z(x)$ is
\begin{equation}
z(x)=\frac{1}{1-e^{-x}}.
\end{equation}
 Also, $Z({\mathcal A},{\mathcal B})=\prod_{a\in {\mathcal A}\atop b\in {\mathcal B}}z(a+b)$,
with the dagger on $Z^\dagger(S,S^-)$ imposing the additional
restriction that a factor $z(x)$ is omitted if its argument is
zero.  The set $\mathcal{A}-S$ comprises the elements in $\mathcal{A}$ but not in $S$ and the notation $|T|$ denotes the number of elements in set $T$. 

\end{theorem}

There is a convention to call the terms in \eqref{eq:J*} ``zero-swap" terms where $|S|=|T|=0$, ``one-swap" where $|S|=|T|=1$, and so forth. This is due to the fact that in the original formulation of the moments and ratios theorems these terms had the appearance of swapping the role of some number, say $n$, of the $a$'s with $n$ of the $b$'s and so the terms now designated with $|S|=|T|=n$ became referred to as  $n$-swap terms.  This terminology has caught on, so we will continue this convention here. 

 The trace of $U^j$ for some positive integer $j$ is $\Tr(U^j)=e^{ij\theta_1}+e^{ij\theta_2}+\cdots+e^{ij\theta_N}$.  The trace of the conjugate transpose of $U^j$ is $\Tr((U^{\dagger})^j)=e^{-ij\theta_1}+e^{-ij\theta_2}+\cdots+e^{-ij\theta_N}$. The relation between the logarithmic derivative and the traces looks like this (if $|y|<1$)
\begin{eqnarray}\label{eq:traces}
y\frac{\Lambda_U'(y)}{\Lambda_U(y)}&=&y\sum_{n=1}^N\frac{-e^{-i\theta_n}}{1-ye^{-i\theta_n}} \nonumber \\
&=&- y\sum_{n=1}^Ne^{-i\theta_n}\sum_{j=0}^\infty (ye^{-i\theta_n})^j \nonumber \\
&=&-\sum_{j=1}^{\infty}y^j\sum_{n=1}^N e^{-ij\theta_n}\nonumber \\
&=&-\sum_{j=1}^\infty y^j \Tr((U^\dagger)^j),
\end{eqnarray}
meaning that we can connect averages of products of traces of $U^\dagger$ to averages of the logarithmic derivative of the characteristic polynomial of $U$:

\begin{eqnarray}\label{eq:traces-as-coeff}
&&\int_{U(N)} \prod_{i=1}^n\left(-A_i\frac{\Lambda_U'}{\Lambda_U}(A_i) \right)\prod_{j=1}^m\left(-B_j\frac{\Lambda_{U^\dagger}'}{\Lambda_{U^\dagger}}(B_j)\right)dU \nonumber \\
&&=\sum_{i_1,\ldots,i_n=1}^{\infty} \sum_{j_1,\ldots,j_m=1}^\infty A_1^{i_1}\cdots A_n^{i_n}B_1^{j_1}\cdots B_m^{j_m} \int_{U(N)} \Tr((U^\dagger)^{i_1})\cdots \Tr((U^\dagger)^{i_n})\Tr(U^{j_1})\cdots\Tr(U^{j_m}) dU.
\end{eqnarray}

We will see a simple case of this in the next section.

\section{Second moment of the logarithmic derivative and two-trace averages}

Averages of two traces are completely and simply understood \cite{kn:dyson3,kn:diasha,kn:rains97}, but we obtain the known results here using the Ratios Theorem as a demonstration of the simplest case of this method.

Let $a,b>0$.  Then from Theorem \ref{theo:J}, 

\begin{eqnarray}\label{eq:1-1case}
&&\int_{U(N)} e^{-a}\frac{\Lambda_U'}{\Lambda_U}(e^{-a})e^{-b}\frac{\Lambda_{U^\dagger}'}{\Lambda_{U^\dagger}}(e^{-b})dU \nonumber \\
&&=\left(\frac{z'}{z}\right)' (a+b)+e^{-N(a+b)} z(a+b)z(-a-b),
\end{eqnarray}
where 
\begin{equation}
z(x)=\frac{1}{1-e^{-x}}.
\end{equation}
For use in this and later calculations, if $e^{-x}=X$, $|X|<1$, note that
\begin{eqnarray}\label{eq:zeds}
z(x)z(-x)&=&\frac{-e^{-x}}{(1-e^{-x})^2}=-\frac{X}{(1-X)^2}= -\sum_{n=0}^\infty nX^n\nonumber \\
\frac{z'}{z}(x)&=&\frac{-e^{-x}}{1-e^{-x}}=-\frac{X}{(1-X)}=-\sum_{n=0}^\infty X^{n+1}{\rm \;\;\;and\;\;\;}\nonumber \\
\left(\frac{z'}{z}\right)' (x)&=&\frac{e^{-x}}{(1-e^{-x})^2}=\frac{X}{(1-X)^2}= \sum_{n=0}^\infty nX^n.
\end{eqnarray}
On the other hand, if $e^{-x}=X$, $|X|>1$,
\begin{equation}\label{eq:zeds_gt_1}
\frac{z'}{z}(x)=\frac{-e^{-x}}{1-e^{-x}}=-\frac{X}{(1-X)}=\frac{1}{1-1/X}=\sum_{n=0}^\infty X^{-n}.
\end{equation}
Returning to \eqref{eq:1-1case}, for convenience we will write $e^{-a}=A$, $e^{-b}=B$, and so we have
\begin{eqnarray}
&&\int_{U(N)} A\frac{\Lambda_U'}{\Lambda_U}(A) B\frac{\Lambda'_{U^\dagger}}{\Lambda_{U^\dagger}}(B) dU\nonumber \\
&&=\frac{AB}{(1-AB)^2}-\frac{A^{N+1}B^{N+1}}{(1-AB)^2},
\end{eqnarray}
where the first term is a zero-swap term and the second is a one-swap term. 

Since $|AB|<1$, we can  write
\begin{eqnarray}\label{eq:2fromratios}
&&\int_{U(N)} A\frac{\Lambda_U'}{\Lambda_U}(A) B\frac{\Lambda'_{U^\dagger}}{\Lambda_{U^\dagger}}(B) dU\nonumber \\
&&=\sum_{n=0}^\infty n (AB)^n-\sum_{j=0}^\infty j(AB)^{N+j}\nonumber \\
&&=\sum_{n=1}^N n (AB)^n+\sum_{n=N+1}^\infty n (AB)^n-\sum_{j=N+1}^\infty (j-N)(AB)^{j}\nonumber \\
&&=\sum_{n=1}^N n(AB)^n+N\sum_{j=N+1}^{\infty}(AB)^j .
\end{eqnarray}

Now, we relate this to traces.  Using \eqref{eq:traces-as-coeff}, 
\begin{eqnarray}\label{eq:secondmoment}
&&\int_{U(N)} A\frac{\Lambda'}{\Lambda}(A) B\frac{\Lambda'}{\Lambda}(B) dU\nonumber \\
&&=\sum_{m,k=1}^{\infty} A^{m}B^{k}\int_{U(N)} \Tr((U^\dagger)^m)\Tr(U^k) dU.
\end{eqnarray}

Comparing \eqref{eq:secondmoment} with \eqref{eq:2fromratios}, we can read off that
\begin{equation}\label{eq:tr2a}
\int_{U(N)}\Tr(U^j)\Tr((U^\dagger)^m)dU_{Haar}=0{\rm \;\; if\;\;} j\neq m
\end{equation}
and
\begin{equation}\label{eq:tr2b}
\int_{U(N)}\Tr(U^j)\Tr((U^\dagger)^j)dU_{Haar}=j{\rm \;\; if\;\;} j<N
\end{equation}
(which is a case proven by Diaconis and Shahshahani \cite{kn:diasha}) and 
\begin{equation}\label{eq:tr2c}
\int_{U(N)}\Tr(U^j)\Tr((U^\dagger)^j)dU_{Haar}=N{\rm \;\; if\;\;} j\geq N,
\end{equation}
which is the result due to Rains (see just before Theorem 1.2 in \cite{kn:rains97}).  All of the above were known to Dyson \cite{kn:dyson3}. Notice that only the zero-swap terms contribute to the Diaconis and Shahshahani range of trace averages $j<N$ (because of the exponent $N$ on $A$ and $B$ in the one-swap term) but not {\em all} of the zero-swap contributions belong to this range - some contribute beyond the $j=N$ transition.

\section{Averages of traces with one $U^\dagger$ and any number of $U$'s}
We now analyse the structure of Theorem \ref{theo:J}  in the case of one $U^\dagger$ trace and any number of traces of $U$.  This restriction greatly simplifies the outer sum in \eqref{eq:J*}. 


We will now prove the following theorem:  
\begin{theorem}\label{thrm: 1--n trace moment}[One dagger trace moment formula]
Let $n\geq2$ and $\beta_1,\ldots,\beta_n\geq1$ be integers. Then
\begin{align}
     \int_{U(N)}\Tr\big((U^\dagger)^\alpha\big) \prod_{j=1}^n \Tr\big(U^{\beta_j}\big)\,dU
     =
     \begin{cases}
     \displaystyle
         \sum_{\varnothing\neq T\subseteq \{1,\ldots,n\}}(-1)^{|T|}
         \left( \sum_{j\in T} \beta_j-N \right)_+,
         & \alpha = \sum_{j=1}^n \beta_j,\\
         0, & \text{otherwise},
     \end{cases}
\end{align}
where the notation $(x )_+$ is zero if the expression $x$ is negative and is equal to $x$ otherwise. 
\end{theorem}
\begin{proof}
The proof proceeds by expanding  $J^*(\mathcal A;\mathcal B)$ from the logarithmic derivatives  theorem, Theorem \ref{theo:J},   in the case $|\mathcal A|=1$ and $|\mathcal B|=n$.  So $\mathcal A=\{a\}$ and $\mathcal B=\{b_1,\ldots,b_n\}$. We will denote this with the shorthand $J^*(\mathcal A_1; \mathcal B_n)$, expand it as a power series in $A=e^{-a}$ and $B_j=e^{-b_j}$, and extract the coefficient of the monomial $A^\alpha\prod_{j=1}^n B_j^{\beta_j}$.  We will require that $|A|,|B_1|,|B_2|,\ldots,|B_n|<1$ and that $|B_n|<|B_{n-1}|<\cdots<|B_1|$.

The outer sum in \eqref{eq:J*}  for this case now only contains the zero-swap and one-swap terms.
Firstly when $|S|=|T|=0$ we have $S=T=\varnothing$. Here we are only left with the inner sum, as all of the $\Zprod$s are omitted, so the contribution is
\begin{equation}\label{eq: S=phi}
\sum_{\mathcal A+\mathcal B \atop{=U_1+\cdots +U_R\atop |U_r|\leq 2}}\prod_{r=1}^{R}H_{\varnothing,\varnothing}(U_r).
\end{equation}
In this case, if $n\geq 2$, we get a contribution of zero because every allowed partition of $\mathcal A+\mathcal B$ will contain a singleton (a single element from $\mathcal A$ or $\mathcal B$). We can see from \eqref{eqn:Hrmt} that for any singleton, $H_{\varnothing,\varnothing}=0$ since we are summing over $\varnothing$. Therefore, there is no zero-swap term when $n\geq 2$. 

The one-swap term is when $|S|=|T|=1$. Here we have $S=\{a\}$ so the outer sum in \eqref{eq:J*} runs over the cases $T=\{b_1\}, T=\{b_2\}, \ldots,T=\{b_n\}$. In the inner sum there is only one partition that leads to a nonzero contribution: the partition of single elements. This leads to 
\begin{eqnarray}
J^*(\mathcal A_1;\mathcal B_n)&=&\sum_{i=1}^{n} e^{-N(a+b_i)} \frac{Z(a,b_i)Z(-a,-b_i)}{Z^{\dagger}(a,-a)Z^{\dagger}(b_i,-b_i)}\prod_{j\in \{1,\ldots,n\}\setminus\{i\} }
H_{\{a\},\{b_i\}}(\{b_j\})\nonumber\\
&=& \sum_{i=1}^{n} e^{-N(a+b_i)}z(a+b_i)z(-a-b_i)\prod_{j\in \{1,\ldots,n\}\setminus\{i\} }\biggl(\frac{z'}{z}(b_j-b_i)-\frac{z'}{z}(a+b_j)\biggr ),
\end{eqnarray}
where the product is over the numbers 1 to $n$, excluding $i$. 

We can use \eqref{eq:zeds} and \eqref{eq:zeds_gt_1} and change into the variables $A=e^{-a}$, $B_j=e^{-b_j}$ to get
\begin{align}
J^*(\mathcal A_1;\mathcal B_n)&=\sum_{i=1}^{n}A^NB_i^N\frac{-AB_i}{(1-AB_i)^2}\prod_{j\in \{1,\ldots,n\}\setminus\{i\} }\biggl(\frac{-B_j/B_i}{1-B_j/B_i}-\frac{-AB_j}{1-AB_j}\biggr )\nonumber\\&=\sum_{i=1}^{n}A^NB_i^N\frac{-AB_i}{(1-AB_i)^2}\prod_{1\leq p<i}\biggl(\frac{1}{1-B_i/B_p}+\frac{AB_p}{1-AB_p}\biggr)\prod_{i<q\leq n }\biggl(\frac{-B_q/B_i}{1-B_q/B_i}+\frac{AB_q}{1-AB_q}\biggr).
\end{align}
Now, since we have $AB_i$ and $B_i/B_j$ with $i>j$ having modulus less than 1, we can expand each term as a power series in analogy to \eqref{eq:zeds} and \eqref{eq:zeds_gt_1}  to obtain the following,
\begin{align}
J^*(\mathcal A_1;\mathcal B_n)=&\sum_{i=1}^{n}\sum_{k=0}^\infty (-k)A^{N+k}B_i^{N+k}\prod_{1\leq p<i}\biggl(\sum_{l_p=0}^\infty\bigl( B_i^{l_p}B_p^{-l_p} +A^{l_p+1}B_p^{l_p+1}\bigr)\biggr)\nonumber\\
&\times\prod_{i<q\leq n }\biggl(\sum_{m_q=0}^\infty \bigl(-B_q^{m_q+1}B_i^{-m_q-1}+A^{m_q+1}B_q^{m_q+1}\bigr)\biggr ).
\end{align}
Each factor inside the products contains two terms. When we expand the products into one sum, each term corresponds to one choice from each factor. We encode these choices using binary vectors $\sigma = (\sigma_p)_{p<i}\in\{0,1\}^{(i-1)}$  and $\tau = (\tau_q)_{q>i}\in\{0,1\}^{(n-i)}$, where
\begin{align}
\begin{aligned}\label{eq: binary choice vectors}
\sigma_p &=
\begin{cases}
0 &\text{if we choose } B_i^{l_p}B_p^{-l_p},\\
1 &\text{if we choose } A^{l_p+1}B_p^{l_p+1},
\end{cases}
\\[4pt]
\tau_q &=
\begin{cases}
0 &\text{if we choose } -B_q^{m_q+1}B_i^{-m_q-1},\\
1 &\text{if we choose } A^{m_q+1}B_q^{m_q+1}.
\end{cases}
\end{aligned}
\end{align}
Each complete selection of terms corresponds uniquely to a pair \((\sigma,\tau)\),
and summing over all such pairs enumerates every term in the full product expansion.
Hence, the products over \(1 \le p < i\) and \(i < q \le n\) can be rewritten as a sum over all
vectors \((\sigma,\tau)\) to obtain
\begin{align}\label{eq: J-expanded-binary}
&J^*(\mathcal A_1;\mathcal B_n)=\sum_{i=1}^{n}\sum_{k=0}^\infty\sum_{l_1,\ldots,l_{i-1}=0}^\infty\sum_{m_{i+1},\ldots,m_{n}=0}^\infty \sum_{\sigma\in\{0,1\}^{i-1}}\sum_{\tau\in\{0,1\}^{n-i}} (-k)(-1)^{(n-i)-\sum_{q>i}\tau_q}A^{E_A}\prod_{r=1}^n B_r^{E_{B_r}}.
\end{align}
Here for the $A$ exponent we write
\begin{align}
& E_A=N+k+\sum_{p=1}^{i-1}\sigma_p(l_p+1)+\sum_{q=i+1}^n\tau_q(m_q+1),
\end{align}
and for the $B$ exponents we have:
\begin{align}
&\nonumber\text{for $r=i:$}\\
&\qquad E_{B_i}=N+k+\sum_{p=1}^{i-1}(1-\sigma_p)l_p-\sum_{q=i+1}^n(1-\tau_q)(m_q+1),\label{eq: E_B_i}\\
&\nonumber\text{for $r =p<i:$}\\
&\qquad E_{B_p} =
\begin{cases}\label{eq: E_B_p}
      -l_p & \text{if } \sigma_p = 0 \\
      l_p+1 & \sigma_p = 1,
\end{cases}\\
&\nonumber\text{for $r =q>i:$}\\
&\qquad E_{B_q} = m_q+1, \qquad (\text{independent of } \tau_q).
\end{align}
We aim to extract the coefficient of a monomial
\begin{align}
    A^{\alpha}\prod_{r=1}^n B_r^{\beta_r}
    \quad\text{with}\quad \alpha\geq1,\ \beta_r\geq1.
\end{align}
From \eqref{eq:traces-as-coeff}, we see that there must be an overall coefficient of zero on any term in the expansion in which some $B_r$ appears with exponent $\leq0$. From the exponent formula for $p<i$, \eqref{eq: E_B_p}, we see that choosing $\sigma_p=0$ forces $E_{B_p}=-l_p\leq0$, hence such terms cannot contribute to the coefficient of $B_p^{\beta_p}$ with $\beta_p\geq1$. Therefore, for every $p<i$ we must have $\sigma_p=1$. With this pruning, the remaining $B$-exponents simplify to
\begin{align}
    E_{B_p}=l_p+1\quad(p<i),\qquad E_{B_q}=m_q+1 \quad (q>i).
\end{align}
Therefore, matching $E_{B_r}=\beta_r$ forces the summation indices to be
\begin{align}
    l_p=\beta_p-1 \quad (p<i), \qquad m_q=\beta_q-1 \quad (q>i),
\end{align}
which are automatically non-negative because $\beta_r\geq1$. Consequently, once $i$ and the target exponents $\beta_1,\ldots,\beta_n$ are fixed, all $l_p$ and $m_q$ are determined, and the only remaining free parameters in \eqref{eq: J-expanded-binary} are $k$ and $\tau=(\tau_q)_{q>i}\in\{0,1\}^{n-i}$. Substituting these values of $l_p$ and $m_q$ into $E_A$ and $E_{B_i}$ gives:
\begin{align}\label{eq: E_A}
    &E_A= \alpha =N+k+\sum_{p=1}^{i-1}\beta_p+\sum_{q=i+1}^n\tau_q\beta_q
\end{align}
and
\begin{align}\label{eq:E_Bequation}
    &E_{B_i}= \beta_i =N+k-\sum_{q=i+1}^n(1-\tau_q)\beta_q.
\end{align}

We now also can see from \eqref{eq: E_A} and \eqref{eq:E_Bequation} that $\alpha = \sum_{r=1}^{n}\beta_r$, as expected. 

We now obtain an expression for $k$ in terms of the $\beta$s. Rearranging our equation for $\beta_i$ we get
\begin{align}\label{eq:kexpression}
    &k= \beta_i -N+\sum_{q>i}(1-\tau_q)\beta_q.
\end{align}
So for a fixed $i$ and choice of $\beta$s, each $\tau\in\{0,1\}^{n-i}$ gives exactly one candidate for $k$. However, this candidate must also satisfy $k\geq1$.

Therefore the contribution from the $i$th term to the coefficient of $A^{\alpha}\prod_{r=1}^n B_r^{\beta_r}$ can be read off from \eqref{eq: J-expanded-binary}, and using the positivity notation defined in the statement of the theorem and the expression for $k$ in \eqref{eq:kexpression},  we then have the coefficient of the monomial $A^{\alpha}\prod_{r=1}^n B_r^{\beta_r}$ is
\begin{align}\label{eq: C_i inequality}
   \sum_{i=1}^n\sum_{\tau\in\{0,1\}^{n-i}}(-1)^{1+(n-i)-\sum_{q>i}\tau_q}\biggl(\beta_i -N+\sum_{q>i}(1-\tau_q)\beta_q\biggr)_+.
\end{align}
Now instead of summing over $\tau$ we can instead think of the subsets of $\{i+1,\ldots,n\}$ where we select all $j$ ($i+1\leq j\leq n$) for which $\tau_j$ is zero.  If we include the sum over $i$, we find that we in fact have a sum over all possible non-empty subsets of $\{1,\ldots,n\}$, as we can think of the role of $i$ in the double sum as indicating the smallest element in each set.  The cardinality of such a set is exactly $1+(n-i)-\sum_{q>i}\tau_q$.

Denoting such a subset by $\Gamma$, and comparing coefficients in  \eqref{eq:traces-as-coeff}, we then get
\begin{align}
    \int_{U(N)}    \Tr\bigl((U^\dagger)^{\alpha}\bigr)    \prod_{j=1}^{n}\Tr\bigl(U^{\beta_j}\bigr)
    \,dU=
\begin{cases} 
      \sum_{\varnothing\neq \Gamma\subseteq\{1,...,n\}}(-1)^{|\Gamma|}\left(\sum_{r\in \Gamma}\beta_r-N \right)_+, & \text{if } \alpha = \sum_{r=1}^{n}\beta_r \\
      0, & otherwise.
\end{cases}
\end{align}

\end{proof}

\section{Negative values of averages of traces with one $U^\dagger$ and three $U$'s}

While all averages involving just 2 or 3 traces are non-negative, numerically it can be seen that averages with one power of the trace of $U^\dagger$ and three powers of the trace of $U$ (called the 1--3 trace average for short, or equivalently the 3--1 average), can achieve negative values.  We will use Theorem \ref{thrm: 1--n trace moment} to describe the region in which these negative values arise.  

Unitary matrices weighted by Haar measure can be generated numerically by following \cite{kn:mezzadri07}.  We show an example in Figure \ref{fig:heatmap}, where values of 
\begin{align}\label{eq:numerical}
  \int_{U(N)}   \Tr\bigl((U^\dagger)^{\beta_1+\beta_2+\beta_3}\bigr)\prod_{j=1}^3 \Tr(U^{\beta_j})dU
\end{align}
are plotted.  Here the exponent on the trace of $U^{\dagger}$, which we sometimes call $\alpha$, is forced to be $\beta_1+\beta_2+\beta_3$ by \eqref{eq:balancing} if we want the average to be something other than zero. 

 A clear region of negativity (depicted in blue)  is visible when all three parameters $\beta_1,\beta_2,\beta_3$ are small relative to $N$. This negative region is not isolated to individual lattice points, but instead forms a connected subset of the $(\beta_1,\beta_2)$-plane whose shape evolves in a systematic way as $\beta_3$ varies. We see that the region of negativity disappears once $\beta_3$ reaches the $N$ barrier, which is consistent with it being bounded by $N$ in both $\beta_1$ and $\beta_2$ when it does occur. It is also bounded by lines in the $(\beta_1,\beta_2)$-plane of the form $\beta_1+\beta_2=k$ for some $k$, where in the $\beta_3=2$ and $\beta_3=3$ cases we have two of these lines forming an upper and lower boundary.

\begin{figure}[]
    \centering
    \includegraphics[width=\linewidth]{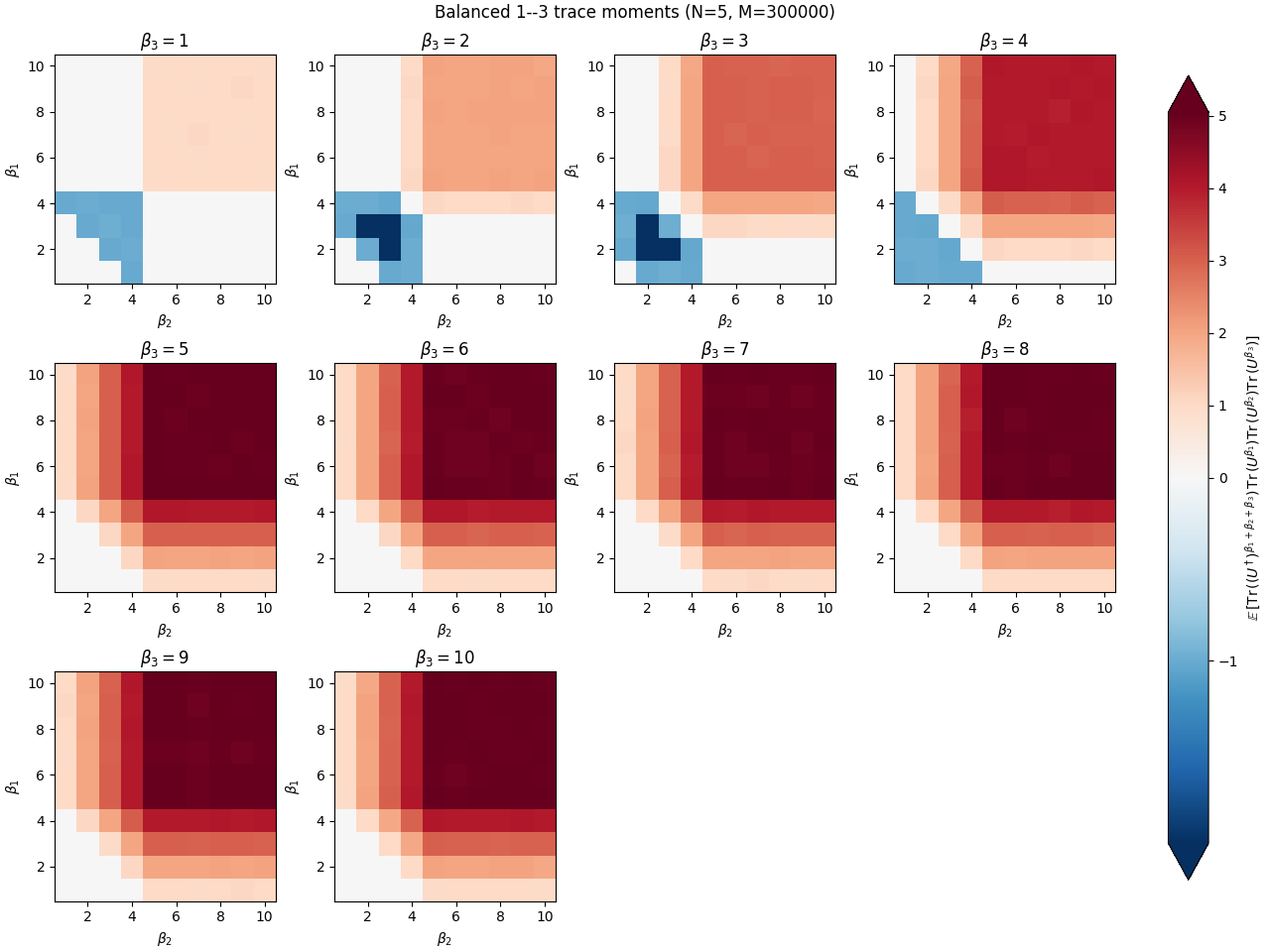}
    \caption{Numerical values of trace averages \eqref{eq:numerical} for matrices from $U(5)$. Blue indicates parameter values where the average of traces takes a negative value. }
    \label{fig:heatmap}
\end{figure}

The subset-sum formula, see Theorem \ref{thrm: 1--n trace moment}, expresses the 1--3 trace average as a linear combination of positive-part terms:
\begin{align}
   \int_{U(N)}\Tr\big((U^\dagger)^\alpha)\prod_{j=1}^3\Tr(U^{\beta_j})dU =
   \begin{cases}
       C(\beta_1,\beta_2,\beta_3;N), & \alpha=\beta_1+\beta_2+\beta_3\\
       0\ & \alpha\neq \beta_1+\beta_2+\beta_3,
   \end{cases}
   \label{eq:1-3-trace-moment}
\end{align}
where
\begin{align}
    C(\beta_1,\beta_2,\beta_3;N)
    &= -(\beta_1+\beta_2+\beta_3-N)_++(\beta_1+\beta_2-N)_++(\beta_1+\beta_3-N)_+ + (\beta_2+\beta_3-N)_+\nonumber\\
    &\qquad - (\beta_1-N)_+-(\beta_2-N)_+ -(\beta_3-N)_+.
    \label{eq:1-3-coefficient}
\end{align}
 Each term of the form $(\cdot)_+$ switches on or off across a plane in $(\beta_1,\beta_2,\beta_3)$--space. Consequently, the trace moment is a piecewise linear function whose sign is determined by the relative position of $(\beta_1,\beta_2,\beta_3)$ with respect to these hyperplanes. Moreover, the condition for a non-zero average,  $\alpha=\beta_1+\beta_2+\beta_3$, removes $\alpha$ as an independent variable. The problem therefore, reduces to understanding a function of three integer variables, making a geometric analysis natural. For notational convenience, we write
\begin{align}
    (x,y,z):=(\beta_1,\beta_2,\beta_3),
\end{align}
and define
\begin{align}
    F(x,y,z):=
    \sum_{\varnothing\neq S\subset\{1,2,3\}}(-1)^{|S|}\left( \sum_{j\in S}\beta_j-N \right)_+.
\end{align}
We now determine precisely where $F(x,y,z)$ can be negative, keeping in mind that in the trace averages $x,y,z>0$.

\paragraph{Step 1: localisation to the cube via singleton terms.}
If $x\geq N$, then every positive part term involving $x$ is active, and the inclusion-exclusion cancelations reduce $F(x,y,z)$ to the 1--2 expansion in the remaining variables
\begin{align}
    F(x,y,z)=(y+z-N)_+-(y-N)_+-(z-N)_+\;\geq\;0.
\end{align}
The same holds if $y\geq N$ or $z\geq N$. Hence any negativity must occur within the cube
\begin{align}\label{eq:neg cube}
    1\leq x,y,z \leq N-1.
\end{align}
In this region, the singleton terms vanish, and so $F$ simplifies to
\begin{align}\label{eq:Fxyz-cube}
    F(x,y,z)=-(x+y+z-N)_++(x+y-N)_++(x+z-N)_++(y+z-N)_+.
\end{align}

\paragraph{Step 2: the lower plane $x+y+z=N$.}
Inside the cube \eqref{eq:neg cube}, if $x+y+z\leq N$ then $(x+y+z-N)_+=0$. Moreover, all pair terms vanish, and so
\begin{align}
    x+y+z\leq N \quad\implies \quad F(x,y,z)=0.
\end{align}
Therefore, negativity (or any nonzero value) requires crossing the lower diagonal plane
\begin{align}
x+y+z\geq N+1.
\end{align}
In this region, we have $(x+y+z)_+=x+y+z-N$, so \eqref{eq:Fxyz-cube} becomes
\begin{align}\label{eq:Fxyz-pairs}
    F(x,y,z)=-(x+y+z-N)+(x+y-N)_++(x+z-N)_++(y+z-N)_+.
\end{align}

\paragraph{Step 3: the upper plane $x+y+z=2N$.}
Still within the cube \eqref{eq:neg cube}, suppose that $x+y+z\geq2N$. Then for each pair,
\begin{align}
    x+y=(x+y+z)-z\geq2N-(N-1)=N+1,
\end{align}
and similarly for $x+z$ and $y+z$. Hence, all three pair terms in \eqref{eq:Fxyz-pairs} are active, and we obtain
\begin{align}
    F(x,y,z)
    &= -(x+y+z)+(x+y-N)+(x+z-N)+(y+z-N)\\
    &= x+y+z-2N\;\geq\;0.
\end{align}
Therefore, negativity cannot occur above the upper diagonal plane
\begin{align}
    x+y+z\geq 2N \quad \implies \quad F(x,y,z)\geq0.
\end{align}

\paragraph{Step 4: negativity in the slab between the planes.}
Assume now that we are inside the cube \eqref{eq:neg cube} and in the slab
\begin{align}\label{eq:slab}
    N+1\leq x+y+z \leq 2N-1.
\end{align}
We claim that $F(x,y,z)<0$ throughout this region. We consider how many of the three pair sums $x+y$, $x+z$, $y+z$ exceed $N$.
\paragraph{(i) All three pair sums exceed $N$.}
In this case, all of the pair sums are activated, and so
\begin{align}
    F(x,y,z)=-(x+y+z-N)+(x+y-N)+(x+z-N)+(y+z-N)=x+y+z-2N.
\end{align}
and so $F<0$ below the upper plane $x+y+z=2N$.
\paragraph{(ii) Exactly two pair sums exceed $N$.}
The two active pair sums must share a common variable, say $x$, without loss of generality. Then $(x+y-N)_+=x+y-N$ and $(x+z-N)_+=x+z-N$ while $(y+z-N)_+=0$, hence
\begin{align}
    F=-(x+y+z-N)+(x+y-N)+(x+z-N)=x-N\leq-1,
\end{align}
since $x\leq N-1$ in the cube.
\paragraph{(iii) Exactly one pair sum exceeds $N$.}
Say $(x+y-N)_+=x+y-N$ while the other two pair terms vanish. Then
\begin{align}
    F=-(x+y+z-N)+(x+y-N)=-z\leq-1.
\end{align}
In each case $F(x,y,z)<0$ and therefore the negative region is precisely:
\begin{align}
     \int_{U(N)}\Tr\big((U^\dagger)^\alpha)\prod_{j=1}^3\Tr(U^{\beta_j})dU<0\iff 1\leq \beta_1,\beta_2,\beta_3 \leq N-1
    \ \text{ and }\
    N+1\leq \beta_1+\beta_2+\beta_3\leq 2N-1.
\end{align}
This matches the blue region in Figure \ref{fig:heatmap}.


\section{Averages of traces with two $U$'s and two $U^\dagger$s  }

In Theorem \ref{theo:J}, we write $e^{-a_1}=A_1$, $e^{-a_2}=A_2$, $e^{-b_1}=B_1$, and $e^{-b_2}=B_2$ and we will require that $|A_1|,|A_2|,|B_1|,|B_2|<1$, as well as $|A_2|<|A_1|$ and $|B_2|<|B_1|$.  Then the theorem gives us:
\begin{eqnarray}\label{eq:sums2:2}
&&\int_{U(N)} A_1\frac{\Lambda_U'}{\Lambda_U}(A_1) A_2\frac{\Lambda_{U}'}{\Lambda_{U}}(A_2) B_1\frac{\Lambda_{U^\dagger}'}{\Lambda_{U^\dagger}}(B_1)B_2\frac{\Lambda_{U^\dagger}'}{\Lambda_{U^\dagger}}(B_2)dU\nonumber \\
&&=\frac{A_1B_1}{(1-A_1B_1)^2}\frac{A_2B_2}{(1-A_2B_2)^2}+\frac{A_1B_2}{(1-A_1B_2)^2}\frac{A_2B_1}{(1-A_2B_1)^2}\nonumber \\
&&+A_1^NB_1^N\frac{-A_1B_1}{(1-A_1B_1)^2}\left[\left(\frac{-B_2/B_1}{1-B_2/B_1}+\frac{A_1B_2}{1-A_1B_2}\right)\left(\frac{-A_2/A_1}{1-A_2/A_1}+\frac{A_2B_1}{1-A_2B_1}\right)+\frac{A_2B_2}{(1-A_2B_2)^2}\right]\nonumber \\
&&+A_1^NB_2^N\frac{-A_1B_2}{(1-A_1B_2)^2}\left[\left(\frac{1}{1-B_2/B_1}+\frac{A_1B_1}{1-A_1B_1}\right)\left(\frac{-A_2/A_1}{1-A_2/A_1}+\frac{A_2B_2}{1-A_2B_2}\right)+\frac{A_2B_1}{(1-A_2B_1)^2}\right]\nonumber \\
&&+A_2^NB_1^N\frac{-A_2B_1}{(1-A_2B_1)^2}\left[\left(\frac{-B_2/B_1}{1-B_2/B_1}+\frac{A_2B_2}{1-A_2B_2}\right)\left(\frac{1}{1-A_2/A_1}+\frac{A_1B_1}{1-A_1B_1}\right)+\frac{A_1B_2}{(1-A_1B_2)^2}\right]\nonumber \\
&&+A_2^NB_2^N\frac{-A_2B_2}{(1-A_2B_2)^2}\left[\left(\frac{1}{1-B_2/B_1}+\frac{A_2B_1}{1-A_2B_1}\right)\left(\frac{1}{1-A_2/A_1}+\frac{A_1B_2}{1-A_1B_2}\right)+\frac{A_1B_1}{(1-A_1B_1)^2}\right]\nonumber \\
&&+A_1^NA_2^NB_1^NB_2^N\frac{-A_1B_1}{(1-A_1B_1)^2}\frac{-A_1B_2}{(1-A_1B_2)^2}\frac{-A_2B_1}{(1-A_2B_1)^2}\frac{-A_2B_2}{(1-A_2B_2)^2}\frac{(1-A_2/A_1)^2}{-A_2/A_1}\frac{(1-B_2/B_1)^2}{-B_2/B_1}.
\end{eqnarray}

The first line after the equal sign is the zero-swap term, the four similar lines in the middle are one-swap terms, and the final line is a two-swap term.  We will see that the two-swap term becomes combinatorially much more complicated and this limits using Theorem \ref{theo:J} for progressively more complicated trace averages. 

The zero- and one-swap terms in this case were explored in \cite{kn:kuplal}, but here we adopt the notation that unifies them with all the other results in this paper.  Identifying the  two-swap terms is new and they are crucial for producing Theorem \ref{theo:2-2full} that holds for unrestricted $N$.

The contributions from the different swap terms are evaluated in the sections below, but here we will state the theorem:
\begin{theorem}\label{theo:2-2full}
Given integers $\alpha_1,\alpha_2,\beta_1,\beta_2$, let $\alpha_1+\alpha_2=\beta_1+\beta_2$. Then the following formula holds for any $N$:
\begin{eqnarray}\label{eq:2-2finalanswer}
&&\int_{U(N)}\Tr\big((U^\dagger)^{\alpha_1})\Tr\big((U^\dagger)^{\alpha_2})\Tr(U^{\beta_1})\Tr(U^{\beta_2})dU\nonumber \\
&&=\beta_1\beta_2(\delta(\alpha_1=\beta_1,\alpha_2=\beta_2)+\delta(\alpha_1=\beta_2,\alpha_2=\beta_1))\nonumber \\
&&-(\beta_1+\beta_2-N)_+ +(\beta_2-N)_++(\beta_1-N)_+\nonumber \\
&&\qquad\qquad +(\beta_1+\beta_2-\alpha_2-N)_+-(\beta_2-\alpha_2-N)_+-(\beta_1-\alpha_2-N)_+\nonumber \\
&&\qquad\qquad +(\beta_1+\beta_2-\alpha_1-N)_+-(\beta_2-\alpha_1-N)_+-(\beta_1-\alpha_1-N)_+\nonumber \\
&&\qquad\qquad -\beta_2(\beta_1-N)_+(\delta(\alpha_1=\beta_1,\alpha_2=\beta_2) +\delta(\alpha_1=\beta_2,\alpha_2=\beta_1))\nonumber \\
&&\qquad\qquad-\beta_1(\beta_2-N)_+(\delta(\alpha_1=\beta_1,\alpha_2=\beta_2)+\delta(\alpha_1=\beta_2,\alpha_2=\beta_1))\nonumber \\
&&  +  \mathbf{1}_{\beta_2> N+1,\alpha_2>N+1}(\min\{\alpha_1,\beta_1\}-N)_+
\sum_{r=-1}^1(-2)^{1-|r|}(|\alpha_1-\beta_1|-r)_+(\beta_2-(\alpha_1-\beta_1)_+-N-1+r)_+ \nonumber
\\  &&\qquad\qquad+  \begin{cases}
        0 & \alpha_1<\beta_2\ \text{or}\ \alpha_2\leq N\ \text{or}\ \beta_2\leq N
    \\  (\max\{\alpha_2,\beta_2\}-N-2)_+ & \alpha_1=\beta_2\ \text{and one or both of}\ \alpha_2 \ \text{or}\ \beta_2 \ \text{is}\ N+1
    \\  (\alpha_2-N)(\beta_2-N) & \alpha_1=\beta_2\ \text{and}\ \alpha_2>N+1\ \text{and}\ \beta_2>N+1
    \\  \alpha_1-\beta_2-1 & \alpha_1>\beta_2\ \text{and}\ \alpha_2=\beta_2=N+1
    \\  -2 & \alpha_1>\beta_2\ \text{and}\ \min\{\alpha_2,\beta_2\}=N+1\
    \\  &  \text{and}\ \max\{\alpha_2,\beta_2\}>N+1
    \\  2(\alpha_1-\beta_2) & \alpha_1>\beta_2\ \text{and}\ \alpha_2>N+1\ \text{and}\ \beta_2>N+1.
    \end{cases}
\end{eqnarray}
Here the delta functions and the indicator function are 1 if their conditions are true and zero otherwise.  The positivity notation $(x)_+$ is zero if $x\leq 0$ and equal to $x$ otherwise.  
\end{theorem}

The following table compares \eqref{eq:2-2finalanswer} to a numerical evaluation, following \cite{kn:mezzadri07}, of the various trace averages with two powers of $\Tr U$ and two $\Tr(U^\dagger)$

\begin{table}[htbp]
\begin{center}
\begin{tabular}{||c|c|c|c|c|c||}
\hline
$\alpha_1$ & $\alpha_2$  & $\beta_1$& $\beta_2$& eq. \eqref{eq:2-2finalanswer}& numerical \\ \hline
7 & 3 & 6& 4 & -2& -1.9935856434522423\\ \hline
6 & 4 & 6 & 4 & 22&21.900323991811124\\ \hline
7& 3 & 3& 7 & 19&18.999789291700733\\ \hline
5 & 16 & 11& 10 & 0& 
0.0715881683791633\\ \hline
10 & 9 & 9 & 10 & 57& 
56.9645409083255\\ \hline
9 & 9 & 9 & 9 & 120& 
119.93166891085507\\ \hline
$\:\:\:$10$\:\:\:$ & $\:\:\:$11 $\:\:\:$& $\:\:\:$11 $\:\:\:$& $\:\:\:$10 $\:\:\:$& 57& 56.912338270268656\\ \hline
12 & 11 & 13 & 10 & -5 &-4.970612272769896\\ \hline
12 & 9 & 12 & 9 & 59 &59.22595730480779\\ \hline
11 & 10 & 12 & 9 & -5 &-5.040295147397936\\ \hline
13 & 10 & 12 & 11 & -5 &-4.982131539971617\\ \hline
\end{tabular}
\caption{Averages of traces calculated by numerically generating matrices from $U(8)$ with Haar measure, compared to \eqref{eq:2-2finalanswer}, for various exponents $\alpha_1,\alpha_2,\beta_1$ and $\beta_2$.}
\label{fig:real}
\end{center}
\end{table}


\subsection{Zero-swap terms}

With the usual method of expanding as geometric series, the zero-swap term becomes
\begin{align}
&\frac{A_1B_1}{(1-A_1B_1)^2}\frac{A_2B_2}{(1-A_2B_2)^2}+\frac{A_1B_2}{(1-A_1B_2)^2}\frac{A_2B_1}{(1-A_2B_1)^2}\nonumber \\
   & =\sum_{n=0}^\infty n(A_1B_1)^n\sum_{m=0}^\infty m(A_2B_2)^m+\sum_{n=0}^\infty n(A_1B_2)^n\sum_{m=0}^\infty m(A_2B_1)^m\nonumber\\
  & = \sum_{n,m=0}^\infty nmA_1^nA_2^mB_1^nB_2^m + \sum_{n,m=0}^\infty nmA_1^nA_2^mB_1^mB_2^n.
\end{align}
In line with \eqref{eq:traces-as-coeff}, we look for the coefficient of $A_1^{\alpha_1}A_2^{\alpha_2}B_1^{\beta_1}B_2^{\beta_2}$, and find that the contribution to the trace average $ \int_{U(N)}\Tr\big((U^\dagger)^{\alpha_1})\Tr\big((U^\dagger)^{\alpha_2})\Tr(U^{\beta_1})\Tr(U^{\beta_2})dU$ from the zero-swap term is
\begin{equation}
\beta_1\beta_2\left[\delta(\alpha_1=\beta_1,\alpha_2=\beta_2)+\delta(\alpha_1=\beta_2,\alpha_2=\beta_1)\right],
\end{equation}
where the delta functions are 1 when all the conditions are satisfied and zero otherwise. 

\subsection{One-swap terms}

The first one-swap term is 
\begin{align}\label{eq:first1swap}
    &A_1^{N} B_1^{N}
\frac{-A_1 B_1}{(1 - A_1 B_1)^2}
\Bigg[
\left(
\frac{-B_2/B_1}{1 - B_2/B_1}
+
\frac{A_1 B_2}{1 - A_1 B_2}
\right)
\left(
\frac{-A_2/A_1}{1 - A_2/A_1}
+
\frac{A_2 B_1}{1 - A_2 B_1}
\right)
+
\frac{A_2 B_2}{(1 - A_2 B_2)^2}
\Bigg]\nonumber\\
&= A_1^NB_1^N\Big(-\sum^\infty_{n=1}n(A_1B_1)^n\Big)\Bigg[\sum^\infty_{m=1}\Big((A_1B_2)^{m}-(B_2/B_1)^{m}\Big)\sum^\infty_{j=1}\Big((A_2B_1)^{j}-(A_2/A_1)^{j}\Big)\nonumber\\
&\qquad\qquad\qquad+\sum_{m=1}^\infty m(A_2B_2)^m\Bigg]\\
&=-\sum_{n,m,j=1}^\infty n\Big(A_1^{N+n+m}A_2^{j}B_1^{N+n+j}B_2^{m}-A_1^{N+n+m-j}A_2^{j}B_1^{N+n}B_2^{m}\nonumber\\
&\qquad\qquad-A_1^{N+n}A_2^{j}B_1^{N+n+j-m}B_2^{m}+A_1^{N+n-j}A_2^{j}B_1^{N+n-m}B_2^{m}\Big)\nonumber\\
&\qquad\qquad\qquad-\sum^\infty_{n,m=1}nmA_1^{N+n}A_2^mB_1^{N+n}B_2^m.\nonumber
\end{align}
Again, we look for the coefficient of $A_1^{\alpha_1}A_2^{\alpha_2}B_1^{\beta_1}B_2^{\beta_2}$.  Once the exponents $\alpha_1, \alpha_2, \beta_1, \beta_2$ are fixed, each sum in the final equality above has at most one term with these powers. In the first sum, for example, the variables of summation must satisfy the following:
\begin{equation}
\alpha_1=N+n+m, \quad \alpha_2 =j,\quad \beta_1=N+n+j,\quad \beta_2=m,
\end{equation}
which has exactly one integer solution:
\begin{equation}
n=\beta_1-\alpha_2-N,\quad j=\alpha_2, \quad m=\beta_2.
\end{equation}
However, this is only an allowable solution if all of $n, j$ and $m$ are positive.  Hence, the contribution of the first sum, in the positivity notation, is
\begin{equation}\label{eq:ambiguity}
-(\beta_1-\alpha_2-N)_+.
\end{equation}
Note that due to the condition for a non-zero average, $\alpha_1+\alpha_2=\beta_1+\beta_2$, there is an ambiguity over how to write \eqref{eq:ambiguity}, and we could equally well have written $-(\alpha_1-\beta_2-N)_+$.  

The other sums in \eqref{eq:first1swap} are dealt with in exactly the same manner.  The total contribution of the first one-swap term in \eqref{eq:sums2:2}  is
\begin{equation}
-(\beta_1-\alpha_2-N)_+ +(\beta_1-N)_++(\beta_1+\beta_2-\alpha_2-N)_+-(\beta_1+\beta_2-N)_+-\delta(\alpha_1=\beta_1,\alpha_2=\beta_2)\beta_2(\beta_1-N)_+.
\end{equation}

The second one-swap term is
\begin{eqnarray}\label{eq:second1swap}
&&A_1^NB_2^N\frac{-A_1B_2}{(1-A_1B_2)^2}\left[\left(\frac{1}{1-B_2/B_1}+\frac{A_1B_1}{1-A_1B_1}\right)\left(\frac{-A_2/A_1}{1-A_2/A_1}+\frac{A_2B_2}{1-A_2B_2}\right)+\frac{A_2B_1}{(1-A_2B_1)^2}\right]\nonumber\\
&&=-\sum_{n,m,j=1}^\infty n\Big(A_1^{N+n}A_2^{j}B_1^{-m+1}B_2^{N+n+m+j-1}-A_1^{N+n-j}A_2^{j}B_1^{-m+1}B_2^{N+n+m-1}\\
&&\quad\qquad+A_1^{N+n+m}A_2^{j}B_1^{m}B_2^{N+n+j}-A_1^{N+n+m-j}A_2^{j}B_1^{m}B_2^{N+n}\Big)\nonumber\\
&&\qquad\qquad\qquad-\sum^\infty_{n,m=1}nmA_1^{N+n}A_2^mB_1^{m}B_2^{N+n}\text{.}\nonumber
\end{eqnarray}
We see that the first two sums in \eqref{eq:second1swap} cannot contribute to the coefficient of $A_1^{\alpha_1}A_2^{\alpha_2}B_1^{\beta_1}B_2^{\beta_2}$ with positive $\beta_1$ because of the exponent $-m+1$ on $B_1$. Similarly to the first one-swap term, the sums on the penultimate and ultimate lines of \eqref{eq:second1swap} contribute 
\begin{equation}
-(\beta_2-\alpha_2-N)_++(\beta_2-N)_+ -\delta(\alpha_1=\beta_2,\alpha_2=\beta_1)\beta_1(\beta_2-N)_+
\end{equation}
to the coefficient of the $A_1^{\alpha_1}A_2^{\alpha_2}B_1^{\beta_1}B_2^{\beta_2}$ term. 

The third one-swap term in \eqref{eq:sums2:2} is
\begin{eqnarray}\label{eq:third1swap}
&&A_2^NB_1^N\frac{-A_2B_1}{(1-A_2B_1)^2}\left[\left(\frac{-B_2/B_1}{1-B_2/B_1}+\frac{A_2B_2}{1-A_2B_2}\right)\left(\frac{1}{1-A_2/A_1}+\frac{A_1B_1}{1-A_1B_1}\right)+\frac{A_1B_2}{(1-A_1B_2)^2}\right]\nonumber \\
&&=-\sum_{n,m,j=1}^\infty n\Big(A_1^{-j+1}A_2^{N+n+m+j-1}B_1^{N+n}B_2^{m}+A_1^{j}A_2^{N+n+m}B_1^{N+n+j}B_2^{m}\\
&&\qquad\qquad-A_1^{-j+1}A_2^{N+n+j-1}B_1^{N+n-m}B_2^{m}-A_1^{j}A_2^{N+n}B_1^{N+n-m+j}B_2^{m}\Big)\nonumber\\
&&\qquad\qquad\qquad-\sum^\infty_{n,m=1}nmA_1^{m}A_2^{N+n}B_1^{N+n}B_2^{m}\text{.}\nonumber
\end{eqnarray}
Two of the sums have a negative exponent on $A_1$ and so must cancel overall.  The remaining terms contribute
\begin{equation}
-(\beta_1-\alpha_1-N)_++(\beta_1+\beta_2-\alpha_1-N)_+-\delta(\alpha_1=\beta_2,\alpha_2=\beta_1)\beta_2(\beta_1-N)_+.
\end{equation}

The fourth one-swap term is
\begin{eqnarray}\label{eq:fourth1swap}
&&A_2^NB_2^N\frac{-A_2B_2}{(1-A_2B_2)^2}\left[\left(\frac{1}{1-B_2/B_1}+\frac{A_2B_1}{1-A_2B_1}\right)\left(\frac{1}{1-A_2/A_1}+\frac{A_1B_2}{1-A_1B_2}\right)+\frac{A_1B_1}{(1-A_1B_1)^2}\right]\nonumber \\
&&=-\sum_{n,m,j=1}^\infty n\Big(A_1^{-j+1}A_2^{N+n+j-1}B_1^{-m+1}B_2^{N+n+m-1}+A_1^{j}A_2^{N+n}B_1^{-m+1}B_2^{N+n+m+j-1}\\
&&\qquad\qquad+A_1^{-j+1}A_2^{N+n+m+j-1}B_1^{m}B_2^{N+n}+A_1^{j}A_2^{N+n+m}B_1^{m}B_2^{N+n+j}\Big)\nonumber\\
&&\quad\quad\quad-\sum^\infty_{n,m=1}nmA_1^{m}A_2^{N+n}B_1^{m}B_2^{N+n}\text{.}\nonumber
\end{eqnarray}
Here only the final two sums have no negative exponents, and these contribute
\begin{equation}
-(\beta_2-\alpha_1-N)_+-\delta(\alpha_1=\beta_1,\alpha_2=\beta_2)\beta_1(\beta_2-N)_+.
\end{equation}
So we find that the contribution to the trace average $ \int_{U(N)}\Tr\big((U^\dagger)^{\alpha_1})\Tr\big((U^\dagger)^{\alpha_2})\Tr(U^{\beta_1})\Tr(U^{\beta_2})dU$ from the one-swap terms is
\begin{eqnarray}\label{eq:2-2oneswap1}
&&-(\beta_1+\beta_2-N)_+ +(\beta_2-N)_++(\beta_1-N)_+\nonumber \\
&&\qquad+(\beta_1+\beta_2-\alpha_2-N)_+-(\beta_2-\alpha_2-N)_+-(\beta_1-\alpha_2-N)_+\nonumber \\
&&\qquad+(\beta_1+\beta_2-\alpha_1-N)_+-(\beta_2-\alpha_1-N)_+-(\beta_1-\alpha_1-N)_+\nonumber \\
&&\qquad-\left[\beta_2(\beta_1-N)_++\beta_1(\beta_2-N)_+\right]\left[\delta(\alpha_1=\beta_1,\alpha_2=\beta_2)+\delta(\alpha_1=\beta_2,\alpha_2=\beta_1)\right],
\end{eqnarray}
or, to write it in a way that might suggest a generalisation, 
\begin{eqnarray}\label{eq:2-2oneswap2}
&&-\left[\beta_2(\beta_1-N)_++\beta_1(\beta_2-N)_+\right]\left[\delta(\alpha_1=\beta_1,\alpha_2=\beta_2)+\delta(\alpha_1=\beta_2,\alpha_2=\beta_1)\right]\nonumber \\
&&\qquad+\sum_{ \substack{S \subset \{1,2\},T\subseteq \{1,2\}\\|S|<2\\0<|T|}}(-1)^{|T|+|S|+1}(\sum_{t\in T}\beta_t- \sum_{s\in S}\alpha_s-N)_{+},
\end{eqnarray}
where  $|T|$ denotes the number of elements in the set $T$.

\subsection{Two-swap terms}


The two-swap term in \eqref{eq:sums2:2} is
\begin{equation}\label{eq:mathcalS}
\mathcal{S}=A_1^NA_2^NB_1^NB_2^N \frac{-A_1B_1}{(1-A_1B_1)^2}\frac{-A_1B_2}{(1-A_1B_2)^2}\frac{-A_2B_1}{(1-A_2B_1)^2}\frac{-A_2B_2}{(1-A_2B_2)^2} \frac{(1-A_2/A_1)^2}{-A_2/A_1} \frac{(1-B_2/B_1)^2}{-B_2/B_1}.
\end{equation}
Expanding as geometric series in the usual way, gives us the rearrangement
\begin{equation}\label{eq:mathcalS2}
    \mathcal{S} = \sum_{n,m,j,k=1}^\infty nmjkA_1^{N+n+m}A_2^{N+j+k}B_1^{N+n+j}B_2^{N+m+k} M,
\end{equation}
where $M$ is the product of the last two factors in \eqref{eq:mathcalS}. We define $I=\{-1,0,1\}$ to get
\begin{align}\label{eq:M}
    M &= (2-A_1^{-1}A_2-A_1A_2^{-1})(2-B_1^{-1}B_2-B_1B_2^{-1})\nonumber
\\  &= \left(\sum_{p\in I}2^{1-|p|}(-1)^pA_1^pA_2^{-p}\right)\left(\sum_{q\in I}2^{1-|q|}(-1)^qB_1^qB_2^{-q}\right)\nonumber
\\  &= \sum_{p,q\in I}{g_{p,q}A_1^pA_2^{-p}B_1^qB_2^{-q}},
\end{align}
where $g_{p,q} = (-2)^{2-|p|-|q|}=2^{2-|p|-|q|}(-1)^{p+q}$.

Now we can plug this in to $\mathcal{S}$ and rearrange to get
\begin{equation}\label{eq:mathcalS3}
\mathcal{S} = \sum_{p,q\in I}\sum_{n,m,j,k=1}^\infty g_{p,q} nmjk A_1^{N+n+m+p}A_2^{N+j+k-p}B_1^{N+n+j+q}B_2^{N+m+k-q}.
\end{equation}
This labelling gives us a systematic way to look for coefficients that match the term $A_1^{\alpha_1}A_2^{\alpha_2}B_1^{\beta_1}B_2^{\beta_2}$. Note that each of $\alpha_1, \alpha_2, \beta_1$ and $\beta_2$ must be greater than or equal to $N+1$. 

Matching exponents gives the equations
\begin{align}
    N+n+m+p &= \alpha_1 \label{eqn:exp1}
\\  N+j+k-p &= \alpha_2 \label{eqn:exp2}
\\  N+n+j+q &= \beta_1  \label{eqn:exp3}
\\  N+m+k-q &= \beta_2. \label{eqn:exp4}
\end{align}
Subtracting (\ref{eqn:exp3}) from  (\ref{eqn:exp1}) and subtracting (\ref{eqn:exp4}) from (\ref{eqn:exp2}) and subtracting (\ref{eqn:exp4}) from (\ref{eqn:exp1}) and rearranging gives (respectively)
\begin{align}
    m &= \alpha_1-N-n-p.
\\  j &= \beta_1-N-n-q.
\\  k &= n-\alpha_1+\beta_2+p+q.
\end{align}
Summing over $n$ (and using the positive part function $(x)_+=\max(x,0)$) all terms in \eqref{eq:mathcalS2} containing $A_1^{\alpha_1}A_2^{\alpha_2}B_1^{\beta_1}B_2^{\beta_2}$ can be written as
\begin{align}
\sum_{p,q\in I}\sum_{n=1}^\infty &g_{p,q} n(\alpha_1-N-n-p)_+(\beta_1-N-n-q)_+
\\\times&(n-\alpha_1+\beta_2+p+q)_+ A_1^{\alpha_1}A_2^{\alpha_2}B_1^{\beta_1}B_2^{\beta_2}. \nonumber
\end{align}

Hence, from \eqref{eq:traces-as-coeff}, the contribution to 
\begin{eqnarray}
&& \int_{U(N)} \Tr((U^\dagger)^{\alpha_1}) \Tr((U^\dagger)^{\alpha_2})\Tr(U^{\beta_1})\Tr(U^{\beta_2}) dU
\end{eqnarray}
from two-swap terms is
\begin{eqnarray}
&& \qquad\qquad\sum_{n=1}^\infty \sum_{p,q\in I} g_{p,q} n (\alpha_1-N-n-p)_+(\beta_1-N-n-q)_+(n-\alpha_1+\beta_2+p+q)_+. \label{eqn:coeff}
\end{eqnarray}
We find that for many values of $n$, the summand turns out to be zero. 

For a given $n$, if any of the three expressions inside the  positive part functions is negative for all $p,q\in I$, then the sum is necessarily zero. 

If all three expressions inside the  positive part functions are non-negative for all $p,q\in I$ (we'll determine for which $n$  this is true later), then we have an expression like
\begin{equation}\label{eq:E}
 \sum_{p,q\in I} g_{p,q} (A-p)(B-q)(C+p+q)
\end{equation}
for some $A$, $B$, and $C$, which are constant relative to $p$ and $q$. We now show that this type of term cancels to zero.

We rearrange \eqref{eq:E} as a quadratic in $p$, so we have
\begin{equation}
 \sum_{p,q\in I}g_{p,q}(B-q)(A(C+q) + Ap-p(C+q)-p^2).
\end{equation}
Now let $p_i=\sum_{p\in I}g_{p,q}p^i$ for $i=0,1,2$, which lets us rewrite the above  as
\begin{equation}\label{eq:EE2}
     \sum_{q\in I}(B-q)(A(C+q)p_0 + (A-C-q)p_1-p_2).
\end{equation}
But we can see that $p_0$ and $p_1$ are zero, because
\begin{align}
    p_0 &= (-2)^{1-|q|}+(-2)^{2-|q|}+(-2)^{1-|q|}=(-2)^{1-|q|}(1-2+1)=0 \label{eq:p0}
\\  p_1 &= (-2)^{1-|q|}(-1) + 0 + (-2)^{1-|q|}(1) = 0,\label{eq:p1}
\end{align}
and also that $p_2 = (-1)^2g_{-1,q}+0+g_{1,q}=2g_{1,q}$, because $g_{-p,q}=g_{p,q}$.
So our sum \eqref{eq:EE2} reduces to
\begin{equation}
 \sum_{q\in I}2g_{1,q}(q-B).
\end{equation}
By the symmetry of $p$ and $q$, this is $0$ in a similar manner to \eqref{eq:p0} and \eqref{eq:p1}.

 The previous conclusions relied on the fact that, for a given $n$, there was no change of sign  in the three expressions in the positive part functions in \eqref{eqn:coeff} as $p$ and $q$ varied.  

For a bracket like $(A-p)_+$ with $p\in\{-1,0,1\}$, we have that the inside takes values $A-1,A,A+1$. We can see that if $A=-1$ or less, then $(A-p)_+$ would be zero for all three values of $p$.  Thus the whole sum in \eqref{eqn:coeff} would go to zero if $\alpha_1-N-n\leq -1$.  If $A=1$ or more, then $(A-p)_+=(A-p)$, so if $\alpha_1-N-n\geq 1$, the sum in \eqref{eqn:coeff} would be zero too, by the argument starting at \eqref{eq:E}, as long as neither of the other expressions in the positivity functions in \eqref{eqn:coeff} change sign as $p$ and $q$ vary.  If $A=0$ then $A-p$ changes sign as $p$ varies in $\{-1,0,1\}$, so we need to consider the case $\alpha_1-N-n=0$ carefully. 

Similarly for the factor $(\beta_1-N-n-q)_+$, if $\beta_1-N-n\leq 1$, the whole sum in \eqref{eqn:coeff} is zero.  If $\beta_1-N-n\geq 1$, then the sum is zero as long as the other positivity functions don't change sign as $p$ and $q$ vary.  We have to pay attention to the case $\beta_1-N-n=0$. 

For the $(C+p+q)_+$ bracket, the inside takes values $C-2, C-1, C,C+1,C+2$. We have that if $C=-2$ or less then all five cases are zero.  So if $n-\alpha_1+\beta_2\leq -2$ then the sum in \eqref{eqn:coeff} goes to zero.  If  $n-\alpha_1+\beta_2\geq 2$ then the sum will go to zero too unless one of the other positivity functions changes sign as $p$ and $q$ vary. It is the cases $n-\alpha_1+\beta_2\in\{-1,0,1\}$ that we need to treat carefully.

We will start with $\alpha_1-N-n=0$. In this case $n=\alpha_1-N$, but since $n$ has to be non-negative, we can use the positive part of this when we write out which $n$ would contribute. Similarly the second bracket contributes for $n=\beta_1-N$, and the last one contributes for $n-\alpha_1+\beta_2\in\{-1,0,1\}$, so $n\in\{\alpha_1-\beta_2-1,\alpha_1-\beta_2,\alpha_1-\beta_2+1\}$. Thus the only terms that potentially contribute to \eqref{eqn:coeff} are
\begin{equation}\label{eq:5n}
    \sum_{\substack{n\in\{(\alpha_1-N)_+,\\(\beta_1-N)_+,\\(\alpha_1-\beta_2-1)_+,\\(\alpha_1-\beta_2)_+,\\(\alpha_1-\beta_2+1)_+\}}} \sum_{p,q\in\{-1,0,1\}} g_{p,q} n (\alpha_1-N-n-p)_+(\beta_1-N-n-q)_+(n-\alpha_1+\beta_2+p+q)_+.
\end{equation}
We have a maximum of 45 terms, but we will see that only a few can be non-zero.  We also note that if any of the $n$ which we sum over are repeated, e.g. if $\alpha_1-N=\alpha_1-\beta_2+1$, then we are in danger of including the term twice.  We will address this below.

Finally, the use of the positive part notation for the $n$ cases is significant. It works because if $\alpha_1-N$ was negative, for example, the summand would be multiplied by $n=(\alpha_1-N)_+=0$ and disappear.

For the $n=(\alpha_1-N)_+$ case, we get (since only the $p=-1$ term is non-zero) the term
\begin{align}
    (\alpha_1-N)_+\sum_{p,q\in I}g_{p,q}(-p)_+&(\beta_1-\alpha_1-q)_+(\beta_2-N+p+q)_+ \label{eqn:terma1} \nonumber
\\  =(\alpha_1-N)_+\sum_{q\in I}g_{1,q}&(\beta_1-\alpha_1-q)_+(\beta_2-N-1+q)_+.
\end{align}

For the $n=(\beta_1-N)_+$ case, we get the term
\begin{equation}
    (\beta_1-N)_+\sum_{p\in I}g_{p,1}(\alpha_1-\beta_1-p)_+(\alpha_2-N-1+p)_+. \label{eqn:termb1}
\end{equation}

 In order for (\ref{eqn:terma1}) to contribute, it needs for $(\beta_1-\alpha_1-q)_+$ to be non-zero for at least some values of $q$, hence it needs to have $\beta_1-\alpha_1+1>0$, and so have $\beta_1\geq\alpha_1$. Similarly for \eqref{eqn:termb1} to contribute, we must have $\alpha_1\geq\beta_1$.

If $\alpha_1\neq\beta_1$ then only one of (\ref{eqn:terma1}) or \eqref{eqn:termb1} will contribute, and if $\alpha_1=\beta_1$ then as $n=N-\alpha_1=N-\beta_1$ we only need to consider one of the expressions anyway because they come from the same value of $n$ in the sum in (\ref{eq:5n}).

Since only one of these terms will show up at a time, we can combine them by replacing $(\alpha_1-N)_+$ and $(\beta_1-N)_+$ by $(\min\{\alpha_1,\beta_1\}-N)_+$. We replace the first positivity bracket in  (\ref{eqn:terma1}) and \eqref{eqn:termb1} with  $(|\alpha_1-\beta_1|-r)_+$.

For the last bracket inside the sum, if $\beta_1\geq\alpha_1$ then it is $(\beta_2-N-1+q)_+$, if the opposite is true, then it is $(\alpha_2-N-1+p)_+$. This behaviour can be expressed (thanks to the balancing condition) as $(\beta_2 - (\alpha_1-\beta_1)_+-N-1+r)_+$. In total, from the cases $n=(\alpha_1-N)_+$ and $n=(\beta_1-N)_+$ we get a combined term of
\begin{equation}
    (\min\{\alpha_1,\beta_1\}-N)_+\sum_{r\in I}g_{r,1}(|\alpha_1-\beta_1|-r)_+(\beta_2-(\alpha_1-\beta_1)_+-N-1+r)_+. \label{eqn:combined}
\end{equation}

This term  only exists when all $\alpha_1$, $\alpha_2$, $\beta_1$, $\beta_2>N$ because of it  being a two-swap term where all $\alpha$'s and $\beta$'s are raised to the power $N$ in \eqref{eq:mathcalS}. We can see this condition in some of the positivity notation in \eqref{eqn:combined}.

Now for the $n=(\alpha_1-\beta_2+r)_+$ cases, which we consider all at once by summing over $r\in I$. Calling this combined contribution $E$, this gives us
\begin{equation}
E=\sum_{r\in I}    (\alpha_1-\beta_2+r)_+\sum_{p,q\in I}g_{p,q}(\beta_2-N-r-p)_+(\alpha_2-N-r-q)_+(p+q+r)_+.
\end{equation}
We can write this expression in terms of $A=\alpha_1-\beta_2,B=\beta_2-N,$ and $C=\alpha_2-N$ as
\begin{equation}
    E=\sum_{p,q,r\in I}  g_{p,q} (A+r)_+(B-r-p)_+(C-r-q)_+(p+q+r)_+.
\end{equation}


If $A<0$ or $B<1$ or $C<1$ then the whole expression is $0$. For example the $A<0$ case means $(A+1)_+,(A)_+,$ and $(A-1)_+$ are all $0$.

We look first at the case $A=0$. Then only the $r=1$ term remains, which becomes
\begin{equation}
    E_{A=0}=\sum_{p,q\in I }g_{p,q}(B-1-p)_+(C-1-q)_+(p+q+1)_+.
\end{equation}
If, in addition, $B=1$ then the only $(B-1-p)_+$ term which survives is the $p=-1$ case, and since $p+q>-1$, $q=1$ is forced and we get $E_{A=0;B=1}=g_{-1,1}(C-2)_+=(C-2)_+$.
Similarly, if $A=0$ and $C=1$ then we get $E_{A=0;C=1}=(B-2)_+$.  In particular, if one of $B$ or $C$ is 1, then the other must be greater than 2 for a non-zero contribution, but in any case we can combine these cases by saying if $A=0$ and either $B$ or $C$ is 1, then the contribution is $(\max\{B,C\}-2)_+$.

If $A=0$ and $B,C>1$ then all of the positive part brackets have non-negative insides, and so we can drop the positive parts and work with the sum as is. We get
\begin{equation}
   E_{A=0;B,C>1} \sum_{p,q\in I }g_{p,q}(B-1-p)(C-1-q)(p+q+1)_+=BC.
\end{equation}

Now we consider the case that $A>0$. Then we have $A-1\geq0$ and so for all $r$, $(A+r)_+=(A+r)$, which allows us to write
\begin{equation}
    E_{A>0}=\sum_{p,q,r\in I}  g_{p,q} (A+r)(B-r-p)_+(C-r-q)_+(p+q+r)_+.
\end{equation}

If $A>0$ and $B=C=1$, then only $r=-1$ can contribute by the following argument.   If $r=0$, since we get contribution only when $1-r-p>0$,  we must have $p<1$. Similarly the $(C-r-q)_+$ factor requires $q<1$ for non-zero contribution. However, the factor $(p+q+r)_+$ is zero if $p<1$ and $q<1$, so there is no contribution from $r=0$.   If $r=1$, then $(B-r-p)_+$  is non-zero when $p<0$, implying  $p=-1$, and similarly from the $(C-r-q)_+$ factor $q<0$, but then $p+q+r<0$, so we get no contribution. If $r=-1$, then $(p+q+r)_+$ is non-zero when $p=q=1$ and so  $E_{A>0;B=C=1}=g_{1,1}(A-1)(1+1-1)(1+1-1)(1+1-1)=(A-1)$.

If $A>0$, $B=1$ and $C>1$, then we write
\begin{equation}
    E_{A>0;B=1;C>1}=\sum_{p,q,r\in I}  g_{p,q} (A+r)(1-r-p)_+(C-r-q)_+(p+q+r)_+.
    \end{equation}
If $r=1$, then the $(1-r-p)_+$ factor is only positive for $p=-1$, and then $(p+q+r)_+$ forces  $q=1$.  Hence we get a contribution of $g_{1,1}(A+1)(1)(C-2)(1)=(A+1)(C-2)$. If $r=0$ then $1-r-p$ could be positive for $p=0,-1$ but since $p+q>0$, by the final factor,  the only viable choice is $p=0$ and $q=1$, and so we get $g_{0,1}A(1)(C-1)(1)=A(-2C+2)$. If $r=-1$ then $(p+q+r)_+$ forces  $p=q=1$ and we get $g_{1,1}(A-1)(1)(C+1-1)(1)=(A-1)C$. Hence the total contribution is $E_{A>0;B=1;C>1}=(A+1)(C-2)+A(-2C+2)+(A-1)C=-2$.  Because of the symmetry, $E_{A>0;B>1;C=1}=-2$ also.

If $A>0$ and  $B,C>1$, then we write
\begin{align}
    E_{A>0;B,C>1} =&\sum_{p,q,r\in I} g_{p,q}(A+r)(B-r-p)(C-r-q)(p+q+r)_+
\\    =& (A+1)[(B-2)C+B(C-2)+4(B-1)(C-1)
\\     & -4(B-2)(C-1)-4(B-1)(C-2)+3(B-2)(C-2)] \nonumber
\\     & +A[-2B(C-1)-2(B-1)C+2(B-1)(C-1)] \nonumber
\\     & +(A-1)BC \nonumber
\\    =& 2A.
\end{align}

So in summary,
\begin{align}
    E=
    \begin{cases}
        0 & A<0\  \text{or}\  B<1\text{or}\ C<1
    \\  (\max\{B,C\}-2)_+ & A=0\ \text{and one or both of}\ B\ \text{or}\ C \ \text{is 1}
        \\  BC & A=0\ \text{and}\ B>1\ \text{and}\ C>1
    \\  A-1 & A>0\ \text{and}\ B=C=1
    \\  -2 & A>0\ \text{and}\ \min\{B,C\}=1,\max\{B,C\}>1
    \\  2A & A>0\ \text{and}\ B>1\ \text{and}\ C>1.
    \end{cases}
\end{align}

Recalling definitions of $A,B$ and $C$, we get the combined contribution from $n=\alpha_1-\beta_2-1,\alpha_1-\beta_2$ and $\alpha_1-\beta_2+1$:
\begin{align}\label{eq:EaabbN}
    E(\alpha_1,\alpha_2,\beta_1,\beta_2,N)=E=
    \begin{cases}
        0 & \alpha_1<\beta_2\ \text{or}\ \alpha_2\leq N\ \text{or}\ \beta_2\leq N
    \\  (\max\{\alpha_2,\beta_2\}-N-2)_+ & \alpha_1=\beta_2\ \text{and one or both of}\ \alpha_2 \ \text{or}\ \beta_2 \ \text{is}\ N+1
    \\  (\alpha_2-N)(\beta_2-N) & \alpha_1=\beta_2\ \text{and}\ \alpha_2>N+1\ \text{and}\ \beta_2>N+1
    \\  \alpha_1-\beta_2-1 & \alpha_1>\beta_2\ \text{and}\ \alpha_2=\beta_2=N+1
    \\  -2 & \alpha_1>\beta_2\ \text{and}\ \min\{\alpha_2,\beta_2\}=N+1\
    \\  &  \text{and}\ \max\{\alpha_2,\beta_2\}>N+1
    \\  2(\alpha_1-\beta_2) & \alpha_1>\beta_2\ \text{and}\ \alpha_2>N+1\ \text{and}\ \beta_2>N+1.
    \end{cases}
\end{align}

When we combine \eqref{eqn:combined} and  \eqref{eq:EaabbN}  we need to account for the fact that if $\alpha_1-N=\alpha_1-\beta_2+1$, for example, we would be double counting this if we summed the (\ref{eqn:combined}) and $E$.  We need to have $\alpha_1-N\neq\alpha_1-\beta_2-r$ for any $-1\leq r\leq1$, which can be rearranged into $N\neq\beta_2+r$, and using $\beta_1-N\neq\alpha_1-\beta_2-r$ and $\alpha_1+\alpha_2=\beta_1+\beta_2$, we similarly get $N\neq\alpha_2+r$ for any $-1\leq r\leq 1$. This amounts to saying that we have duplication in the $n$ values when $\beta_2=N+1$ or $\alpha_2=N+1$, as we never see exponents less than $N+1$ in two-swap terms (see comment after \eqref{eq:mathcalS3}).  This results in the following 2-swap contribution (where the indicator function is 1 when both conditions hold and 0 otherwise)
\begin{gather}
    \mathbf{1}_{\beta_2> N+1,\alpha_2>N+1}(\min\{\alpha_1,\beta_1\}-N)_+
\\  \times\sum_{r\in I}g_{r,1}(|\alpha_1-\beta_1|-r)_+(\beta_2-(\alpha_1-\beta_1)_+-N-1+r)_+ \nonumber
\\  +E(\alpha_1,\alpha_2,\beta_1,\beta_2,N). \nonumber
\end{gather}

\section{General trace moments}

The designation of terms as zero-swap, one-swap, etc, is not just convenient book-keeping.  It was shown in a series of papers \cite{kn:consna14, kn:massna16b,chandee,kn:demsna,kn:chaleeli} to be of critical importance when one studies the $n$-correlation of eigenvalues of random matrices from one of the classical compact groups with Haar measure.  This in turn yields information on the $n$-correlation of zeros of the Riemann zeta function and other $L$-functions through the now familiar number theory/random matrix theory connection (see for example \cite{kn:keasna00a,kn:keasna00b,kn:consna06,kn:consna07,kn:consna08,kn:massna16}). For example, if one investigates
\begin{equation}\label{eq:examplecorr}
 \int_{U(N)} \sum_{j_1,\dots,j_n} F(\theta_{j_1},\dots,\theta_{j_n}) dU
 \end{equation}
 in order to probe the behaviour of $n$-tuples of eigenvalues, the support of the Fourier transform of the test function $F$ plays a large role.  In number theory, the analogous calculations can often only be carried out for test functions whose Fourier transforms have very restricted support.  It turns out that a quantity like \eqref{eq:examplecorr} can be calculated using $J^*$ in Theorem \ref{theo:J}. If the Fourier transform of the test function has support only when the sum of certain Fourier variables is less than $2q-\epsilon$, for some positive $\epsilon$, then {\em only} terms in $J^*$ with $|S|=|T|<q$ contribute.  In a similar vein, we have already seen in our examples that one-swap terms cannot contribute to trace averages where the sum of the $\alpha$ exponents is less than  $N+1$ and two-swap terms cannot contribute to trace averages when the sum of the exponents is less than $2N+2$. So there is benefit in calculating low swap terms separately.  
 
In the following sections, we will address the zero- and one-swap terms for a general trace moment, one of the consequences of which is the following theorem, which approximately doubles the range of Diaconis and Shahshahani's original result:

 \begin{theorem}\label{theo:generalfull}
     Given integers $\alpha_1,\ldots,\alpha_K,\beta_1,\ldots,\beta_L$, for $\alpha_1+\cdots+\alpha_K=\beta_1+\cdots+\beta_L<2N+2$, we have, with $\mathcal K=\{1, \ldots, K\}$ and $\mathcal L=\{1,\ldots, L\}$,
     \begin{eqnarray}
        && \int_{U(N)}\Tr((U^\dagger)^{\alpha_1})\cdots \Tr((U^\dagger)^{\alpha_K})\Tr(U^{\beta_1})\cdots\Tr(U^{\beta_L})dU=\nonumber \\
         &&\delta(K=L)\sum_{({\mathcal K}:{\mathcal L})} \left( \beta_1\cdots \beta_{L} \delta(\alpha_{k_1}=\beta_{\ell_1}, \alpha_{k_2}=\beta_{\ell_2}, \ldots, \alpha_{k_{L}}=\beta_{\ell_{L}})\right)\nonumber \\
         &&+\sum_{\genfrac{}{}{0pt}{1}{R\subsetneq \mathcal K,\, Q \subsetneq \mathcal L}{|R| = |Q|}} \left(\sum_{(R:Q)} \left(\prod_{j=1}^{|{\mathcal Q}|}\beta_{q_j}\right)\left(\delta(\alpha_{r_1}=\beta_{q_1},\ldots,\alpha_{r_{|{\mathcal Q}|}}=\beta_{q_{|{\mathcal Q}|}})\right)\right)\nonumber \\
 &&\times \sum_{ \substack{S \subset R^c,T\subseteq Q^c\\|S|<|R^c|\\0<|T|}}(-1)^{|T|+|R^c|-|S|-1}\left(\sum_{t\in T}\beta_t-\sum_{s\in S}\alpha_s-N\right)_{+}.
     \end{eqnarray}
     Here the notation $({\mathcal K}:{\mathcal L})$, for two sets of the same cardinality,  indicates summing over all possible pairings of one element from  $\mathcal K$ with one element of $\mathcal L$, with the pairs labelled $(a_{k_1},b_{\ell_1})$, $(a_{k_2},b_{\ell_2})\ldots$ $(a_{k_{|\mathcal K|}},b_{\ell_{|\mathcal L|}})$. A $\delta$ function is zero unless all the conditions are true, in which case it is 1. The set $R^c$ is the complement of $R$ in $K$. The positivity notation $(x)_+$ is zero if $x\leq 0$ and $x$ otherwise. 
 \end{theorem}

 \begin{proof}
     The proof follows from the derivation of the zero- and one-swap terms in the following sections, with the additional information that two-swap and higher terms will always have $\alpha_1+\cdots+\alpha_K\geq 2N+2$.  We can see this because any two-swap terms will have factors like \eqref{eq:mathcalS}, where we can see from \eqref{eq:mathcalS2} and \eqref{eq:M} that the two selected $A$'s and $B$'s have exponent at least $N+1$.  In general trace averages, these will be multiplied by various $H_{S,T}(W)$, see \eqref{eqn:Hrmt}, but the contribution of these to the {\em sum} of exponents on all the $A$'s (or $B$'s) is zero or higher, as we can see from \eqref{eq:zeds} and \eqref{eq:zeds_gt_1}. Similarly, for higher swap terms, say $\ell$-swap, there will be $\ell$ $A$'s and $\ell$ $B$'s with exponent at least $N+1$ for the same reason. 
 \end{proof}

 \subsection{Zero-swap terms}

 We examine the zero-swap terms of a general average of traces. This is covered by Kuperberg and Lal{\'i}n in \cite{kn:kuplal}, but we include it here for completeness and to ensure consistency of notation. 
 
 Write a general moment of logarithmic derivatives as
 \begin{eqnarray}\label{eq:generallogderiv}
\int_{U(N)}\prod_{k\in {\mathcal K}}
(-e^{-a_k})\frac{\Lambda_U'}{\Lambda_U}(e^{-a_k})\prod_{\ell\in
{\mathcal L}} (-e^{-b_\ell})\frac{\Lambda_{U^\dagger}'}{\Lambda_{U^\dagger}}(e^{-b_\ell})~ dU,
 \end{eqnarray}
with $\mathcal K=\{1, \ldots, K\}$ and $\mathcal L=\{1,\ldots, L\}$.
 The zero-swap terms in \eqref{eq:generallogderiv}, which we will call ${\mathcal S}_0$, only exist if $\mathcal K$ and $\mathcal L$ have the same cardinality, that is $|\mathcal K|=|\mathcal L|$, so that in \eqref{eqn:Hrmt} we can pair every $a$ with a $b$.  So we have
 \begin{eqnarray}
 {\mathcal S}_0=\mathbf{1}_{|\mathcal K|=|\mathcal L|}  \sum_{({\mathcal K}:{\mathcal L})} \left(\prod_{j=1}^{|{\mathcal L}|} \left(\frac{z'}{z}\right)' (a_{k_j}+b_{\ell_j}) \right),
 \end{eqnarray}
 where the notation $({\mathcal K}:{\mathcal L})$ indicates summing over all possible pairings of one element from  $\mathcal K$ with one element of $\mathcal L$, with the pairs labelled $(a_{k_1},b_{\ell_1})$, $(a_{k_2},b_{\ell_2})\ldots$ etc, and the indicator function is zero if $ |\mathcal K|\neq|\mathcal L|$.
 
  With $ |\mathcal K|=|\mathcal L|$, in the usual manner we will let $A_j=e^{-a_j}$ and $B_j=e^{-b_j}$, $|A_j|,|B_j|<1$, for $j= 1,2,\ldots,|\mathcal L|$, and $|A_{|\mathcal L|}|<\cdots<|A_1|$ and $|B_{|\mathcal L|}|<\cdots<|B_1|$.  So
  \begin{eqnarray}
 { \mathcal S_0}&=&\mathbf{1}_{|\mathcal K|=|\mathcal L|}  \sum_{({\mathcal K}:{\mathcal L})} \left(\prod_{j=1}^{|{\mathcal L}|} \frac{A_{k_j}B_{\ell_j}}{(1-A_{k_j}B_{\ell_j})^2}\right)\nonumber \\
 &=&\mathbf{1}_{|\mathcal K|=|\mathcal L|}  \sum_{({\mathcal K}:{\mathcal L})} \left(\prod_{j=1}^{|{\mathcal L}|}\left(\sum_{n_j=1}^\infty n_j(A_{k_j}B_{\ell_j})^{n_j} \right)\right)\nonumber \\
 &=&\mathbf{1}_{|\mathcal K|=|\mathcal L|}  \sum_{({\mathcal K}:{\mathcal L})}\; \sum_{n_1=1}^\infty\cdots \sum_{n_{|\mathcal L|}=1}^\infty n_1\cdots n_{|\mathcal L|}A_{k_1}^{n_1}\cdots A_{k_{|\mathcal L|}}^{n_{|\mathcal L|}}B_{\ell_1}^{n_1}\cdots B_{\ell_{|{\mathcal L}|}}^{n_{|\mathcal L|}}.
 \end{eqnarray}
 
 So, the contribution from zero-swap terms to the trace average 
 \begin{equation} \int_{U(N)}\Tr((U^\dagger)^{\alpha_1})\cdots \Tr((U^\dagger)^{\alpha_K})\Tr(U^{\beta_1})\cdots\Tr(U^{\beta_L})dU
 \end{equation} is
 \begin{eqnarray}\label{eq:generalzeroswap}
 \mathbf{1}_{K=L}\sum_{({\mathcal K}:{\mathcal L})} \left( \beta_1\cdots \beta_{L} \delta(\alpha_{k_1}=\beta_{\ell_1}, \alpha_{k_2}=\beta_{\ell_2}, \ldots, \alpha_{k_{L}}=\beta_{\ell_{L}})\right).
 \end{eqnarray}
 Note that this must coincide with the formula of Diaconis and Shahshahani \eqref{eq:DS} if $\beta_1+\cdots +\beta_L\leq N$, as it is only zero-swap terms that contribute in this range. Diaconis and Shahshahani write their trace averages slightly differently, recording multiplicity in the exponents, \begin{equation}\label{eq:DS1}
    \int_{U(N)} \prod_{j=1}^l ({\rm Tr} M^j)^{a_j}\overline{({\mathrm Tr}M^j)}^{b_j} dU = \delta_{\mathbf{ab}} \prod_{j=1}^l j^{a_j} a_j!
\end{equation}
(note these $a$'s and $b$'s are different from those in the average of logarithmic derivatives elsewhere in this paper). We can see that \eqref{eq:generalzeroswap} and \eqref{eq:DS1} are both zero if the list of exponents on the traces of the $U$'s is not equal to the list of exponents on the daggered traces. The product $\prod_{j=1}^lj^{a_j}$ in \eqref{eq:DS1} is just the product of the list of exponents, and so corresponds to  $ \beta_1\cdots \beta_{L} $.  If the $\alpha$'s are all distinct then we have only one non-zero delta function in the sum \eqref{eq:generalzeroswap} and all the $a_j$ are 1 in \eqref{eq:DS1}. But if there is a set of the $\alpha$'s (and therefore also the $\beta$'s) that are equal with multiplicity $a$, then there are $a!$ ways to assign the duplicated $\alpha$'s to the set of equal $\beta$'s.  Thus there will be $a!$ non-zero delta functions in \eqref{eq:generalzeroswap} that capture these arrangements.  This extends naturally to any number of blocks of identical exponents, and so accounts for the $a_j!$ in \eqref{eq:DS1}. 

 \subsection{One-swap terms}

 The one-swap terms of \eqref{eq:generallogderiv}, which we will call ${\mathcal S_1}$, are written in a slightly different form (merely a rearrangement of \eqref{eq:J*}) in equation (64) of \cite{kn:demsna}. We will use this as our starting point.  Below the sets $R$ and $Q$ select first the pairs containing one $a$ and one $b$ that will contribute to $H_{S,T}(W)$ in \eqref{eqn:Hrmt} and the indices $k$ and $l$ identify the single element of $S$ and single element of $T$ in \eqref{eq:J*}.
 \begin{eqnarray}
 {\mathcal S}_1&=& \sum_{\genfrac{}{}{0pt}{1}{R\subsetneq \mathcal K,\, Q \subsetneq \mathcal L}{|R| = |Q|}} \left(\sum_{(R:Q)} \prod_{j=1}^{|Q|} \left(\frac{z'}{z}\right)' (a_{r_j} +b_{q_j})\right) \nonumber \\
    &&\times 
     \sum_{R_1 \subsetneq R^c,\, Q_1 \subsetneq Q^c} \, \sum_{k \in R_1^c,\, l \in Q_1^c}  e^{-N(a_{k} + b_{l})}z(a_k +b_l)z(-a_k-b_l) \nonumber \\ 
     && \times \left(\prod_{s \in R_1}\frac{z'}{z}(a_s-a_k)\right)
    \left(\prod_{s \in R_1^c\backslash\{k\}} -\frac{z'}{z}(a_s+b_l)
    \right) \nonumber \\
    && \times \left(\prod_{t \in Q_1}\frac{z'}{z}(b_t-b_l)\right)
    \left(\prod_{t \in Q_1^c\backslash\{l\}} -\frac{z'}{z}(b_t+a_k)\right),
 \end{eqnarray}
 with $\mathcal K=\{1, \ldots, K\}$ and $\mathcal L=\{1,\ldots, L\}$.  The complement notation means that $R^c$ is the complement of $R$ in $\mathcal K$.  Similarly $R_1^c$ is the complement of $R_1$ in $R^c$. The set $R_1$ selects the singleton sets $W=\{a\}$ in \eqref{eqn:Hrmt} from which we select the first term of $H_{S,T}(W)$ when we expand the $H$'s, and the set $R_1^c\backslash\{k\}$ selects singleton sets $W=\{a\}$ in \eqref{eqn:Hrmt} from which we select the second term of $H_{S,T}(W)$. Including $k$ in $R_1^c$ is just a convenience used in \cite{kn:demsna} that we continue here.  The notation on the $\mathcal L$ side is the same. 
 
 Once again $A_j=e^{-a_j}$ and $B_j=e^{-b_j}$, $|A_j|,|B_j|<1$, $|A_{K}|<\cdots<|A_1|$ and $|B_{L}|<\cdots<|B_1|$. So
  \begin{eqnarray}\label{eq:bigS}
 {\mathcal S}_1&=& \sum_{\genfrac{}{}{0pt}{1}{R\subsetneq \mathcal K,\, Q \subsetneq \mathcal L}{|R| = |Q|}} \left(\sum_{(R:Q)} \prod_{j=1}^{|Q|} \left(\sum_{n_j=0}^\infty n_jA_{r_j}^{n_j}B_{q_j}^{n_j}\right)\right) \nonumber \\
    &&\times 
     \sum_{R_1 \subsetneq R^c,\, Q_1 \subsetneq Q^c} \, \sum_{k \in R_1^c,\, l \in Q_1^c}  A_k^NB_l^N\left(-\sum_{n=0}^\infty nA_k^nB_l^n \right)\nonumber \\ 
     && \times \left(\prod_{\substack{s \in R_1\\s<k}}\left(\sum_{n_s=0}^\infty A_k^{n_s}A_s^{-n_s}\right)\right)\left(\prod_{\substack{s \in R_1\\s>k}}\left(-\sum_{n_s=0}^\infty A_s^{n_s+1}A_k^{-n_s-1}\right)\right)  \left(\prod_{s \in R_1^c\backslash\{k\}}\left( \sum_{n_s=0}^\infty A_s^{n_s+1}B_l^{n_s+1}\right)
    \right) \nonumber \\
    && \times  \left(\prod_{\substack{t \in Q_1\\t<l}}\left(\sum_{n_t=0}^\infty B_l^{n_t}B_t^{-n_t} \right)\right) \left(\prod_{\substack{t \in Q_1\\t>l}}\left(-\sum_{n_t=0}^\infty B_t^{n_t+1}B_l^{-n_t-1} \right)\right) 
    \left(\prod_{t \in Q_1^c\backslash\{l\}}\left( \sum_{n_t=0}^\infty B_t^{n_t+1}A_k^{n_t+1}\right)\right).
 \end{eqnarray}
There are several things to notice about \eqref{eq:bigS}. We only see terms where the total sum of exponents on all the $A$'s equals the total sum of exponents on all the $B$'s.  And since the $A$'s and $B$'s in $\mathcal R$ and $\mathcal Q$ have pairwise equal exponents, this means that the total sum of exponents on all the $A$'s in $\mathcal R^c$ is equal to the total sum of exponents on all the $B$'s in $\mathcal Q^c$. We notice that if there are any $s\in R_1$ with $s<k$, then the exponent on $A_s$ is non-positive (see the first bracket on the penultimate line of \eqref{eq:bigS}) and so can't contribute to a term of the form $A_1^{\alpha_1}\cdots A_{K}^{\alpha_{K}}B_1^{\beta_1}\cdots B_{L}^{\beta_L}$ where $\alpha_1,\ldots,\alpha_K,\beta_1,\ldots, \beta_L\geq 1$.  Therefore, contributing terms have no $s\in R_1$ with $s<k$ and similarly no $t\in Q_1$ with $t<l$ (from the first bracket on the last line of \eqref{eq:bigS}).  Inside the sum over $R_1$ and $Q_1$, if we seek to match exponents with a term of the form $A_1^{\alpha_1}\cdots A_{K}^{\alpha_{K}}B_1^{\beta_1}\cdots B_{L}^{\beta_L}$, the exponent on $B_l$  will be $\beta_l=N+n-\sum_{t\in Q_1, t>l}\beta_t+\sum_{s\in R_1^c\backslash\{k\}}\alpha_s$ or in other words $n=(\beta_l+\sum_{t\in Q_1, t>l}\beta_t-\sum_{s\in R_1^c\backslash\{k\}}\alpha_s-N)_+$. Note that the combination of $\alpha$'s and $\beta$'s in this expression for $n$ cannot be empty, as it contains at least $\beta_l$, and it cannot include the full set of all $\alpha$'s and $\beta$'s because it cannot include $\alpha_k$. In addition, $l$ is always the smallest index on a $\beta$ in this expression for $n$. We would end up with the same expression for $n$ if we focused instead on the exponent of $A_k$, $\alpha_k=N+n-\sum_{s\in R_1, s>k}\alpha_s+\sum_{t\in Q_1^{c}\backslash\{l\}}\beta_t$, because of the total sum of exponents on all the $A$'s in $\mathcal R^c$ equaling the total sum of exponents on all the $B$'s in $\mathcal Q^c$ mentioned above.

So, the contribution from one-swap terms to the trace average 
 \begin{equation} \int_{U(N)}\Tr((U^\dagger)^{\alpha_1})\cdots \Tr((U^\dagger)^{\alpha_K})\Tr(U^{\beta_1})\cdots\Tr(U^{\beta_L})dU
 \end{equation} is
 \begin{eqnarray}\label{eq:generaloneswap}
 &&\sum_{\genfrac{}{}{0pt}{1}{R\subsetneq \mathcal K,\, Q \subsetneq \mathcal L}{|R| = |Q|}} \left(\sum_{(R:Q)} \left(\prod_{j=1}^{|{\mathcal Q}|}\beta_{q_j}\right)\left(\delta(\alpha_{r_1}=\beta_{q_1},\ldots,\alpha_{r_{|{\mathcal Q}|}}=\beta_{q_{|{\mathcal Q}|}})\right)\right)\nonumber \\
 &&\times \sum_{ \substack{S \subset R^c,T\subseteq Q^c\\|S|<|R^c|\\0<|T|}}(-1)^{|T|+|R^c|-|S|-1}(\sum_{t\in T}\beta_t-\sum_{s\in S}\alpha_s-N)_{+}.
\end{eqnarray}
 The number $|T|+|R^c|-|S|+1$ counts the minus signs in \eqref{eq:bigS} as described below.   This expression agrees with all of the smaller cases in previous sections. 
 
 We arrive at \eqref{eq:generaloneswap} because the set $R_1^c\backslash \{k\}$ in $n=(\beta_l+\sum_{t\in Q_1, t>l}\beta_t-\sum_{s\in R_1^c\backslash\{k\}}\alpha_s-N)_+$ will hit every subset of $R^c$ (except $R^c$ itself) as we sum over $R_1$ and the choice of $k\in R_1^c$ in \eqref{eq:bigS}, so this is represented in \eqref{eq:generaloneswap} by the sum over the subsets $S\subset R^c, |S|<|R^c|$. The sign $(-1)^{|R^c|-|S|-1}$ arises because there is a minus sign (in the middle bracket of the penultimate line of \eqref{eq:bigS}) associated with each $A_i$ where $i\in R^c$ is {\em not} chosen for $S$ and also not equal to $k$.  Similarly,  $\beta_l+\sum_{t\in Q_1, t>l}\beta_t$ can be made to include any non-empty set of $\beta_i$'s, with $i\in Q^c$, by the appropriate choice of $Q_1$ and $l$ in the sum in \eqref{eq:bigS}. This is represented by $T\subseteq Q^c$, $0<|T|$, and the sign $(-1)^{|T|}$ arises because there is a minus sign (in the middle bracket of the last line of \eqref{eq:bigS}) associated with each $B_i$ with $i\in T$. 
 
 It remains to argue that each choice of $S$ and $T$ occurs only once in the sums in \eqref{eq:bigS}. We could imagine, perhaps, the same set of $\beta$'s occurring in different ways, with different $\beta$'s playing the role of $\beta_l$.  However, this cannot occur as $l$ is always the smallest index occurring in the combination $\beta_l+\sum_{t\in Q_1, t>l}\beta_t$. Similarly, $k$ is always the smallest index in $R^c$ that is missing from $R_1^c\backslash \{k\}$, meaning that the choice of sets $S$ and $T$ are in one-to-one correspondence with contributing terms in the sums over $R_1$, $Q_1$, $k$ and $l$ in \eqref{eq:bigS}.



\end{document}